\documentclass[12pt]{article}
\usepackage{amsthm,amsmath,amsfonts,amssymb}
\usepackage{graphicx,psfrag,epsf}
\usepackage{enumerate}
\usepackage{natbib}

\newcommand{\blind}{0}

\RequirePackage[colorlinks,citecolor=blue,urlcolor=blue]{hyperref}
\RequirePackage{graphicx}
\usepackage[caption = true]{subfig}

\usepackage{xcolor}
\usepackage{bm}
\usepackage{enumerate}
\usepackage{tabularx}
\usepackage{float}
\usepackage{wasysym}

\newcommand{\Oxi}{\Omega_{x_i}(f_D)}

\newcommand{\methodsym}{\Omega\left(\tilde{f}_D\right)}
\newcommand{\meshy}{\Omega\left(\tilde{f}_D\right)}

\newcommand{\di}{d_{(i)}}
\newcommand{\diplus}{d_{(i)+1}}

\newcommand{\ej}{\boldsymbol{e}_j}

\newcommand{\method}{MBS}
\newcommand{\Method}{MBS}

\newcommand\bovermat[2]{%
	\makebox[0pt][l]{$\smash{\overbrace{\phantom{%
					\begin{matrix}#2\end{matrix}}}^{\text{#1}}}$}#2}

\newtheorem{theorem}{Theorem}[section]
\newtheorem{lemma}[theorem]{Lemma}

\begin{document}

	\def\spacingset#1{\renewcommand{\baselinestretch}%
		{#1}\small\normalsize} \spacingset{1}

	
	\if0\blind
	{
		\title{\bf Mesh-Based Solutions for Nonparametric Penalized Regression}
		\author{Brayan Ortiz \\
			Modeling \& Optimization, Amazon.com \\
			and \\
			Noah Simon\thanks{
				Noah Simon was supported to do this work by NIH grant DP5OD019820 from the office of the director.}\hspace{.2cm} \\
			Department of Biostatistics, University of Washington}
		\maketitle
	} \fi
	
	\if1\blind
	{
		\bigskip
		\bigskip
		\bigskip
		\begin{center}
			{\LARGE\bf Mesh-Based Solutions for Nonparametric Penalized Regression}
		\end{center}
		\medskip
	} \fi
	
	\bigskip
	\begin{abstract}
		It is often of interest to estimate regression functions non-parametrically. Penalized regression (PR) is one statistically-effective, well-studied solution to this problem. Unfortunately, in many cases, finding exact solutions to PR problems is computationally intractable. In this manuscript, we propose a mesh-based approximate solution (\method) for those scenarios. \method~transforms the complicated functional minimization of NPR, to a finite parameter, discrete convex minimization; and allows us to leverage the tools of modern convex optimization. We show applications of \method~in a number of explicit examples (including both uni- and multi-variate regression), and explore how the number of parameters must increase with our sample-size in order for \method~to maintain the rate-optimality of NPR. We also give an efficient algorithm to minimize the \method~objective while effectively leveraging the sparsity inherent in \method.
	\end{abstract}
	
	\noindent%
	
	\spacingset{1.45}

	\section{Introduction}
	
	In this paper we consider a statistical problem in which we measure a response $y_i$ and a covariate $x_i\in[a,b]$, on each of $i=1,\ldots,N$ observations. We assume a generative model of the form
	\[
	y_i = f^* \left(x_i \right) + w_i
	\]
	where $f^*$ is an unknown function from a known function class $\mathcal{F}$, and $w_i$ are iid errors with $\operatorname{E}\left[w_i\right] = 0$ and $\operatorname{var}{w_i} = \sigma^2 < \infty$. We are interested in estimating $f^*$ based on the observed data. One common approach for estimating $f^*$ is to use penalized regression (\cite{buhlmann2011statistics}, \cite{geer2000empirical}):
	\begin{equation}\label{eq:pen}
	\hat{f} = \operatorname{argmin}_{f\in\mathcal{F}}\frac{1}{n}\sum_{i=1}^n\left(y_i - f\left(x_{i}\right)\right)^2 + \lambda_n P\left(f\right)
	\end{equation}
	where $\lambda_N \geq 0$ is a tuning parameter and $P(\cdot)$ is a penalty function which penalizes ``complexity.'' Some common examples of penalized regression include: smoothing splines, lasso, fused lasso, trend filtering, locally adaptive regression splines and others (\cite{wahba1978smoothing}, \cite{tibshirani1996regression}, \cite{tibshirani2005fusedlasso}, \cite{kim2009ell_1}, \cite{tibshirani2014adaptive}, \cite{mammen1997}). 
	
	Solutions to penalized problems have good theoretical properties: with carefully chosen $\lambda_N$, these penalized-regression-based estimates often converge at minimax (or near minimax) rates (\cite{massart2007concentration},\cite{geer2000empirical}). Additionally,  when $P(\cdot)$ is convex, and $\mathcal{F}$ can be finitely parametrized, $\hat{f}$ can be solved efficiently (in polynomial time) in both theory and practice (\cite{wahba1975smoothing},\cite{wahba1978smoothing}). Even when $\mathcal{F}$ cannot be finitely parametrized, sometimes the solution for a specific set of observed data falls in a calculable finite dimensional subfamily, and so we can efficiently solve~\eqref{eq:pen}. For example, this is the case for smoothing splines and the fused lasso (\cite{wahba1978smoothing}, \cite{tibshirani2014adaptive}).
	
	We propose a computationally tractable framework for approximately solving~\eqref{eq:pen} when the true solution does not fall in a simple finite dimensional subfamily. In our framework, we alter the optimization problem in~\eqref{eq:pen} slightly. We select a mesh of $m$ knots over the domain of $x$ and use the fitted-values at those knots as our optimization parameters: We replace the penalty function with a finite-difference/Riemann approximation; and the fitted values at the data points are approximated by cleverly interpolating between fitted values at knots. We refer to the general approach as \method, or mesh based solution.  
	
	Formally, we first take an $m$-division $D$ of $[a,b]$:
	\[
	D: a=d_1 \le d_2 \le \ldots \le d_{m-1}\le d_m=b.
	\]
	such that all observed $x_i$'s are positioned within this mesh. Using $D$, we can formulate an approximation to our original problem~\eqref{eq:pen}:  
	\begin{equation}\label{eq:app}
	\tilde{f_D} = \operatorname{argmin}_{f_D\in\mathbb{R}^m}\frac{1}{n}\sum_{i=1}^n\left(y_i - \Oxi \right)^2 + \lambda_N P_D\left(f_D\right),
	\end{equation}
	where $P_D(f_D)$ is a approximation to $P(f)$ based on finite-differences/Riemann sums using our fitted values on the mesh, $f_D = \left(f\left(d_1\right),\ldots,f\left(d_m\right)\right)^\top$; and our interpolator $\Omega: \mathbb{R}^m\rightarrow \mathcal{F}$ with $\Oxi\equiv \Omega(f\left(d_1\right),\ldots,f\left(d_m\right); x_i)\in \mathcal{F}$ takes in fitted values on our mesh, and an $x_i$ (potentially not on the mesh), and calculates an interpolated fit at $x_i$. In Section 2, we describe a framework for translating a particular class of $P(f)$ into $P_D(f_D)$ and discuss piecewise polynomial interpolation schemes. At the end of Section 2, we briefly discuss similarities between univariate \method, trend filter, locally adaptive regression splines, and Whittaker smoothing \citep{whittaker1922new}.  
	
	We briefly detail multivariate \method~in Section 3. We describe an alternating direction method of multipliers (ADMM) solver for univariate \method~when $\ell=1$ (Section 4). Using the ADMM solver, we run a simulation study in Section 5 highlighting that for even a modest number of knots the approximation error induced by replacing~\eqref{eq:pen} by~\eqref{eq:app} is smaller than our statistical error. Indeed, as the distance between mesh points converges to $0$, the solution to our problem~\eqref{eq:app} converges to the solution to the original problem \eqref{eq:pen}. In our discussion, we provide theoretical results supporting our findings, as well as considerations on future work. 
	
\section{\method~Optimization} \label{method}

	In this section we discuss in detail our proposal for approximating \eqref{eq:pen}. For now we restrict ourselves such that
	\[
	P(f) = \int \left|f^{(r+1)}(d)\right|^\ell \partial d
	\] 
	for some integer $r\ge 0$ and $\ell>0$. These Sobolev-norm penalties are a fairly broad class, which include smoothing splines and total-variation penalties among others. Let us begin by recalling the $m$-division or mesh, $D$, of $[a,b]$:
	\[
	D: a=d_1\le d_2 \le \ldots \le d_{m-1}\le d_m=b.
	\]
	Let $\boldsymbol{\delta}=(\delta_1,\ldots,\delta_{m-1})$ denote the bin widths within the mesh, where $\delta_j=d_{j+1}-d_j$ for $j=1,\ldots,m-1$. Define $\delta_{max}$ as the max width within $D$ such that $\delta_m=\max\{d_{j+1}-d_{j}|0\le j < m\}$. Often, we specify $D$ as a regular or even mesh, where $\delta_j=d_{j+1}-d_j=\frac{b-a}{m}$ for $j=1,\ldots,m-1$. For regular $D$, $\delta_{max}=\frac{b-a}{m}$. The approximate problem we aim to solve is:
	\begin{equation}\label{eq:method}
	\tilde{f_D} = \operatorname{argmin}_{f_D\in\mathbb{R}^m}\frac{1}{n}\sum_{i=1}^n\left(y_i - \Oxi \right)^2 + \lambda_n P_D\left(f_D\right),
	\end{equation}
	where $P_D(f_D)$ is an approximation to $P$ calculated using only $d_1,\ldots,d_m$ and $f_D=\left(f\left(d_1\right),\ldots,f\left(d_m\right)\right)$. For notation, $f(d_j)=\left(f_D\right)_j$. When an observation $x_i\ne d_j$ (for all $j$), $f_D$ will not be defined and so $\Oxi$ approximates $f\left(x_i\right)$ using only $x_i$, $\left(d_1,\ldots,d_m\right)$ and $f_D$. We must choose $P_D$ and $\Oxi$ to best approximate $P$ and $f(x_i)$, respectively. 
	
	\subsection{Choosing $P_D(f_D)$:} We use finite-differences/Riemann sums to approximate $P$. This works in part because of the form we have assumed for $P$: $P(f) = \int \left|f^{(r+1)}(d)\right|^\ell \partial d$, where $f^{(r)}=\frac{\partial^r}{\partial d^r}f(d)$. The operator $\frac{\partial}{\partial d}$ is only defined for differentiable functions $f\in \mathcal{F}$. Hence, we are motivated to discretize $f$ by its evaluation on a mesh $D$, i.e. $f_D=\left( f(d_1),\ldots, f(d_m) \right)^\top \in \mathbb{R}^m$, since can always take differences to approximate a discrete derivative. We define the normalized first order difference function $\Delta^1_m:\mathbb{R}^m\to \mathbb{R}^{m-1}$ such that 
	\[
	[\Delta^1_m f_D]_i=\frac{f(d_{i+1})-f(d_i)}{\delta_i},
	\]
	where $i=1,\ldots,m-1$. Our Riemann approximation approach is built around these easily calculable estimates of the first-order discrete derivatives of $f_D$, i.e. $[\Delta^1_m f_D]\in \mathbb{R}^{m-1}$. We treat the operator $\Delta^1_m$ as analogous to $\frac{\partial}{\partial d}$. Furthermore, we replace the Riemann integral with the Riemann sum, using the length of the mesh increments, $\delta_i=d_{i+1}-d_i$, in place of $\partial d$. For example, suppose $r=0$ and $\ell=1$. Under our framework, an estimate of $P(f)$ on a regular mesh $D$ is
	\begin{align}
	P_{D}(f_D) & =  \sum_{i=1}^{m-1}\left|[\Delta^1_m f_D]_i\right|\delta_i \\
	&= \sum_{i=1}^{m-1}\left|f\left(d_{i+1}\right) - f\left(d_i\right)\right|,
	\end{align}
	i.e. the fused lasso problem. 
	
	For regular $D$, we define a normalized $r$th order difference function $\Delta^r_m: \mathbb{R}^m\to \mathbb{R}^{m-r}$ such that 
	\[
	\left[\Delta^r_mf_D\right]_i = \left[\Delta^{r-1}_{m-1}\left[\Delta^1_mf_D\right]\right]_i=\ldots = \left[\Delta^1_{m-r+1}\left[\Delta^1_{m-r}\left[\ldots \left[\Delta^1_mf_D\right]\ldots\right]\right]\right]_i,
	\]  
	where $i=1,\ldots,m-r$. For general $r$ and $\ell$, our Riemann approximation to $P(f)$ on a regular mesh takes the following form:
	\begin{gather}
	P_{D}(f_D) = \sum_{i=1}^{m-r} \left|\delta_i^{\frac{1}{\ell}}[\Delta_m^{r+1} f_D]_i\right|^\ell.
	\end{gather}
	
	$P_{D}(f_D)$ takes relatively simple forms, namely iterated differences of differences. For example, with $r=1$ and $\ell \ge 1$, we get
	\[
	P_{D}(f_D) = \left(\frac{b-a}{m}\right)^{1-2\ell}\sum_{i=1}^{m-2}   \left(f\left(d_{i+2}\right) -2f\left(d_{i+1}\right)  + f\left(d_i\right)\right)^\ell.
	\]
	On a general mesh $D$, the algebra becomes more complex, but we still avoid the analytical complexity of assuming differentiability of $f$. The algebraic formulation of $P_D(f_D)$ will be useful in the multivariate case. In the supplement, we present matrix representation for $P_D(f_D)$, for both regular and irregular mesh cases, which will be useful for deriving solvers for univariate \method. Briefly, we show in the supplement that
	\begin{gather}
	P_{D}(f_D) = \left\|\delta^{\frac{1}{\ell}-r-1}\Delta_{m}^{(r+1)} f_{D}\right\|_\ell^\ell,
	\end{gather}    
	where recursively $\Delta^{(r)}_{n} = \Delta^{(r-1)}_{n-1} \cdot \Delta^{(1)}_n$ and 
	\begin{gather*}
	\Delta^{(1)}_n  =
	\begin{pmatrix}
	-1 & 1 & 0 & \ldots  & 0 & 0 \\
	0 & -1 & 1 & \ldots & 0 & 0 \\
	\vdots &  &  &  & & \\
	0 & 0 & 0 & \ldots & -1 & 1 	
	\end{pmatrix} \in \mathbb{R}^{(n-1)\times n}.
	\end{gather*} 
		
	\subsection{ Choosing $\Oxi$:} There's a vast literature on function interpolation/estimation \citep{burden1989numerical}. Popular choices include linear interpolation, higher order piecewise polynomial interpolation and splines. To maintain convexity of \eqref{eq:app} we ensure $\Oxi$ is a linear function of $f_D$ for each $x_i$. We briefly discuss a general approach, then show how this specializes to piecewise polynomial and spline interpolation.
	
	Suppose we would like to interpolate $b$ points $\theta_{1:b} = \left[f\left(d_1\right),\ldots, f\left(d_b\right)\right]^{\top}$, via $\tilde{f} \gets \sum_{i=1}^b \alpha_i\psi_i$ a linear combination of pre-specified basis elements $\psi_1,\ldots,\psi_b$. Consider the design matrix $\Psi$ with $\Psi_{ij} = \psi_j\left(d_i\right)$. We can find our coefficients $\alpha$ by solving the linear system $\Psi\alpha = \theta_{1:b}$, where $\theta_{1:b}=\left(\theta_1,\ldots,\theta_b\right)^\top$. Thus $\hat{\alpha} = \Psi^{-1}\theta_{1:b}$. Now if we would like to evaluate this interpolation at a new point $x$; we first form $\tilde{\psi} = \left[\psi_1(x),\ldots,\psi_b(x)\right]^{\top}$, and then compute $\tilde{\psi}^{\top}\hat{\alpha} = \tilde{\psi}^\top \Psi^{-1}\theta_{1:b}$. Note that this is linear in $\theta_{1:b}$. 
	
	Suppose we interpolate $f_D$ via a $k$th-order natural spline. First, we define the knot superset $T=\{d_{\left \lfloor k/2 \right \rfloor +2},\ldots, d_{m- \left \lfloor k/2 \right \rfloor } \}$. Let $\tilde{d}=\left(\tilde{d}_1,\ldots, \tilde{d}_{m-k-1} \right)^\top$ denote the column vector containing the set $T$. Next, we specify the basis elements, i.e. 
	\[
	\Psi_1(x)=1, \Psi_2(x)=x,\ldots, \Psi_{k+1}(x)=x^k, \mbox{ and } \Psi_{k+1+j'}(x)=(x-\tilde{d}_{j'})^k_+,
	\]
	for $j'=1,\ldots,m-k-1$. Thus, we have a design matrix $\Psi$ with $\Psi_{ij}=\Psi_j(d_i)$, $i,j=1,\ldots,m$ and coefficients given by $\hat{\alpha} = \Psi^{-1}\theta_{1:m}$, where $\theta_{1:m}\equiv f_D$. For $n$ points $\boldsymbol{x}$, we form an $n\times m$ matrix $\tilde{\psi}$ with $\tilde{\psi}_{i,\cdot}= \left[\psi_1(x_i),\ldots,\psi_{m}(x_i)\right]$. Finally, we compute the interpolation for $\boldsymbol{x}$, $\tilde{\psi}\hat{\alpha} = \tilde{\psi}\Psi^{-1}\theta_{1:m}=O\theta_{1:m}$. Note that we call the $n\times m$ matrix $O=\left(\boldsymbol{o}_1,\ldots,\boldsymbol{o}_n\right)^\top$ the 'interpolation matrix,' since each $i$th row, $\boldsymbol{o}_i^\top$, provides the linear combination of $\theta_{1:m}$ to produce a fitted value $\hat{f}\left(x_i\right)$. Indeed, we define $\Omega_{\boldsymbol{x}}(f_D)$ as the interpolated values of the observed data $\boldsymbol{x}$ such that
	\[
	\Omega_{\boldsymbol{x}}(f_D) = Of_D. 
	\]         
	
	An alternative interpolation is to fit each observation from a local polynomial (of order $k$). We take advantage of the fact that $k+1$ points uniquely define a $k$th degree polynomial. However, this type of local interpolation is different from the smoothing spline described before, since it lacks continuous first derivatives globally. For a point $x_i$, we find a neighborhood of $k+1$ local mesh points $\tilde{d}=\left(\tilde{d}_1,\ldots, \tilde{d}_{k+1} \right)^\top$ such that $x_i \in [\tilde{d}_1,\tilde{d}_{k+1}]$. Let $\theta_{1:k+1}=\left(f(\tilde{d}_{1}),\ldots, f(\tilde{d}_{k+1})\right)^\top=\left( \theta_{1},\ldots, \theta_{k+1}\right)^\top$ be the evaluations of $f$ for the mesh neighborhood about $x_i$. We interpolate $\theta_{1:k+1}$ using the basis functions:
	\[
	\psi_1(x)=1, \psi_2(x)=x,\ldots, \mbox{ and } \psi_{k+1}(x)=x^k. 
	\]
	From these basis functions, we see that this local polynomial coincides with the regression splines defined above for $k=0,1$. Similar to the general approach, we consider a design matrix $\Psi\in \mathbb{R}^{(k+1)\times (k+1)}$ with $\Psi_{i'j}=\psi_j(\tilde{d}_{i'})$. Thus, we can calculate coefficients $\hat{\alpha}= \Psi^{-1}\theta_{1:k+1}$. To find the fitted value of $x_i$, we form $\tilde{\boldsymbol{\psi}}=\left(\psi_1(x_i),\ldots,\psi_{k+1}(x_i) \right)^\top$ and compute $\tilde{\boldsymbol{\psi}}^\top\hat{\alpha} = \tilde{\boldsymbol{\psi}}^\top\Psi^{-1}\theta_{1:k+1}=\tilde{o}_i^\top\theta_{1:k+1}$. Note that $\tilde{o}_i$ is specific to $x_i$. Indeed, for a new point $x_j \not\in [\tilde{d}_1,\tilde{d}_{k+1}]$, we move to another set of neighboring mesh points $\tilde{d}'$  which leads to $\tilde{o}_j$. Hence, we call this interpolation scheme the 'moving local polynomial.' As before, for observations $\boldsymbol{x}=(x_1,\ldots,x_n)$, we can write
	\[
	\Omega_{\boldsymbol{x}}(f_D) = Of_D,
	\]
	where $O=\left(\boldsymbol{o}_1,\ldots,\boldsymbol{o}_n\right)^\top$. For $x_i\in [d_j,d_{j+k}]$, we define a column of length $m$ 
	\[
	\boldsymbol{o}_i=\left(0,\ldots,0,\tilde{\boldsymbol{o}}_i,0,\ldots, 0\right)^\top, 
	\]
	where $\tilde{\boldsymbol{o}}_i$ occupies indices $j$ through $j+k$. 
	 
	\subsection{ Writing the Univariate \method~Problem:}
	
	We have introduced $P_D(f_D)$ and $\Oxi$. For response $y$ and data $\boldsymbol{x}$ with related $k$th-order interpolation matrix $O$, the \method~objective~\eqref{eq:app} can be written simply as
	\begin{equation}\label{eq:app2}
	\tilde{f_D} = \operatorname{argmin}_{f_D\in\mathbb{R}^m}\|y - Of_D \|_2^2 + \lambda_n \left\|\delta^{\frac{1}{\ell}-r-1}\Delta_{m}^{(r+1)} f_{D}\right\|_\ell^\ell.
	\end{equation}
	Often, we may write~\eqref{eq:app2} as
	\begin{equation}\label{eq:v2}
	\tilde{f_D} = \operatorname{argmin}_{f_D\in\mathbb{R}^m}\|y - Of_D \|_2^2 + \lambda_n \left\|\Delta_{m}^{(r+1)} f_{D}\right\|_\ell^\ell.
	\end{equation}
	We denote the fitted values of the data as 
	\[
	\tilde{f} = O\tilde{f_D}.
	\]
  
	\begin{figure} 
		\centering
		\subfloat[ $r=0$, $k=0$ ]{\includegraphics[height=1.5in,width=1.9 in]{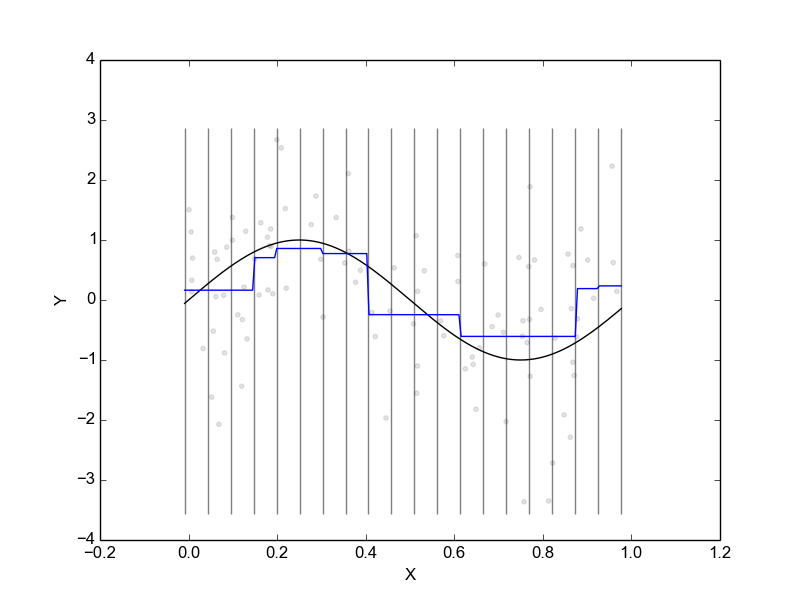}} 
		\subfloat[ $r=1$, $k=1$ ]{\includegraphics[height=1.5in,width=1.9 in]{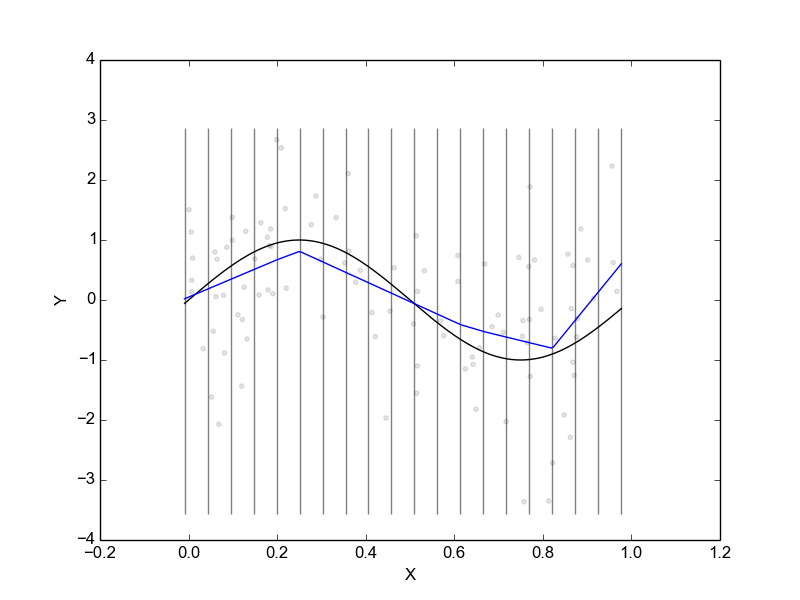}} 
		\subfloat[ $r=2$, $k=2$ ]{\includegraphics[height=1.5in,width=1.9 in]{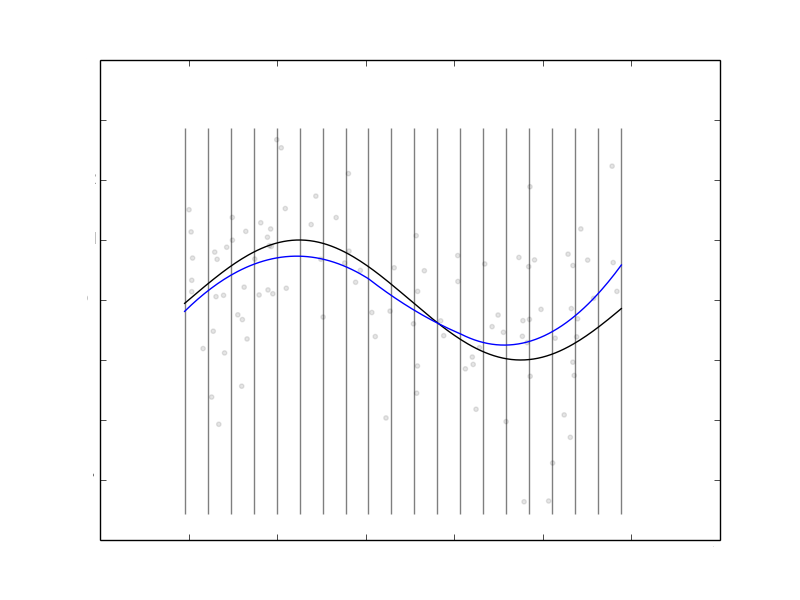}} 
		\caption{We observe $n=100$ noisy observations (transparent) of $f(x)=\sin(2\pi x)$ (black). In blue, we draw \method~estimates with $m=20$ and $(r,k)$ varying. Vertical lines are drawn at the mesh. }
		\label{fig:fits}
	\end{figure}
	In Figure 1, we draw \method~fits using the moving local polynomial over a mesh with $m=20$. Indeed, when $k=r$, we get $r$th-order piecewise polynomial fits. Figure 1 hints to us that \method~will have a relationship with other univariate methods which we explore in the next section.
	
	\subsection{An Alternative Representation of \method~Optimization:} \label{sec:basis}
	Here we show that one can equivalently write our \method~objective~\eqref{eq:method} using a basis expansion. As an alternative to the route we discussed above, one might consider optimizing \eqref{eq:pen} over a linear class $\mathcal{F} \equiv \operatorname{span}\left(\psi_1,\ldots, \psi_K\right)$. As we will discuss in Section~\ref{sec:other}, for some bases $TV\left(\frac{\partial^r}{\partial x^r} \sum\beta_k \psi_k(x)\right)$ has a simple representation, however often it does not. In cases where it does not, we could approximate it using a similar discretization strategy as before: For a given element of our linear space $f = \sum\beta_k \psi_k$ and a mesh $D=(d_1,\ldots, d_m)$ define $f_D\left(\theta\right) \equiv \left[\sum_k\theta_k \psi_k(d_1),\ldots, \sum_k\theta_k \psi_k(d_m)\right]$. Now consider the problem:
	\begin{equation}\label{eq:method_basis}
		\tilde{\theta} = \operatorname{argmin}_{\theta \in \mathbb{R}^K} \frac{1}{n}\sum_{i=1}^n\left(y_i - \sum_k \theta_k \psi_k\left(x_i\right)\right)^2 +\lambda P_D\left(f_D\left(\theta\right)\right)
	\end{equation}
	This is exactly equivalent to our original formulation for \method~\eqref{method} if, for $\psi(D) = \left(\psi(d_1),\ldots \psi(d_m)\right)^\top$, the matrix $\Psi_D = \left(\psi_1(D),\ldots, \psi_K(D)\right)$ is a basis for $\mathbb{R}^m$. In particular, in this case our interpolator is precisely $\Omega_{x}\left(f_D(\theta)\right) = \sum_k \theta_k \psi_k(x)$. If we define $\left(\psi_1,\ldots, \psi_K\right)$ by the rising polynomial basis given in the supplement, then the interpolations are equivalent to those described using the moving local polynomial. However, the moving local polynomial as defined earlier with sparse $\Psi_D$ leads to more efficient matrix operations than the rising polynomial basis with $\Psi_D$ dense. For this computational reason, we prefer implementing interpolation using the moving local polynomial over the rising polynomial basis.     
	
	\section{Comparisons to Other Univariate Methods}\label{sec:other}
	
	There are two other methods for approximately solving~\eqref{eq:pen} with a total variation penalty: $\ell_1$ trend filtering ({\tt TF}) by \cite{tibshirani2014adaptive} and locally adaptive regression splines ({\tt LocARS}) by \cite{mammen1997}. Like the exact solution to the functional problem~\eqref{eq:pen}, both {\tt TF} and {\tt LocARS} give minimax rate-optimal solutions. These two methods use the basis expansion framework discussed in Section~\ref{sec:basis}. They both solve a total variation problem
		\begin{equation}\label{eq:restrict}
		\hat{f}_{TV} \in \operatorname{argmin}_{f\in\mathcal{F}_{n}^{restrict}}\frac{1}{n}\sum_{i=1}^n\left(y_i - f\left(x_{i}\right)\right)^2 + \lambda \cdot TV\left(f^{(r)}\right),
	\end{equation}
	over two different linear subspaces, $\mathcal{F}_{n}^{restrict}$. For {\tt LocARS}, one uses $\mathcal{F}_{n}^{restrict} \equiv \operatorname{span}\left(\psi_1^{LocARS},\ldots, \psi_n^{LocARS}\right)$ where the $\psi_i^{locARS}$ are from the $r$-th order truncated power basis with knots at the observations, $x_i$. For {\tt TF} one uses $\mathcal{F}_{n}^{restrict} \equiv \operatorname{span}\left(\psi_1^{TF},\ldots, \psi_n^{TF}\right)$, where the $\psi_i^{TF}$ are from the $r$-th order falling factorial basis with knots at the observations, $x_i$. More details on these bases can be found in \cite{tibshirani2014adaptive}. These bases are chosen in part because functions in their span permit a simple, finite dimensional, representation of $TV\left(f^{(r)}\right)$. This is in contrast to \method~where the penalty was instead approximated via $P_D(f_D)$. In fact, for each of these bases, \eqref{eq:restrict} can be rewritten as a simple lasso problem:
	\begin{equation}\label{eq:lasso_version}
		\hat{\theta} = \operatorname{argmin}_{\theta\in\mathbb{R}^n}\frac{1}{n}\sum_{i=1}^n\left(y_i - \sum_{k=1}^n \theta_k \psi_k\left(x_{i}\right)\right)^2 + \lambda \sum_{k=r+1}^n \left|\theta_k\right|,
	\end{equation}
	with $\hat{f} \gets \sum_{k=1}^n \hat{\theta}_k \psi_k$, where $\psi_k \equiv \psi_k^{TF}$ or $\psi_k \equiv \psi_k^{locARS}$ for {\tt TF} and {\tt LocARS}, respectively. While the lasso form of these problems is useful for interpretation, solving either {\tt TF} or {\tt LocARS} by applying a general purpose lasso solver to \eqref{eq:lasso_version} is very inefficient, as the design matrix is poorly conditioned. {\tt TF}, unlike {\tt LocARS}, is more amenable to efficient computation: One can further rewrite {\tt TF} as a particular instance of \method:
	\begin{equation}\label{eq:TF}
	\hat{f}_D \in \operatorname{argmin}_{f_D\in\mathbb{R}^n}\frac{1}{n}\sum_{i=1}^n\left(y_i - f_D\left(x_{i}\right)\right)^2 + \lambda P_D\left(f_D\right)
	\end{equation}
	where we use a mesh with grid points at the observations $D=\left(x_1,\ldots,x_n\right)$, and $P_D$ is our discretized TV penalty from before.
	
	For Whittaker Smoothing, the problem can be written similarly to~\eqref{eq:v2} with $m=n$. However, in Whittaker Smoothing, the aim is to adjust a regular mesh of size $n$, not necessarily find the most efficient mesh of size $m<n$ relative to a computationally efficient interpolation matrix $O$. 
	
	There are 2 main downsides to basing the mesh on the observations or the size of the observations: 1) We have $n$ tuning parameters, which can slow down computation, when for statistical accuracy many fewer are needed; and 2) the uneven spacing of the observations means that \eqref{eq:TF} is still a poorly conditioned problem --- this leads to instability for many convex solvers (as noted in \cite{ryan2015fast}). The ability of our method to use a regular mesh, with fewer than $n$ knots is a potentially significant asset (especially in problems with many observations).
	
	\method~has an additional advantage when the features lie in $2$ or more dimensions (discussed further in Section~\ref{sec:multivar}). The obvious extensions of {\tt TF} and {\tt LocARS} to higher dimensions use complicated and computationally inefficient basis expansions in $2+$ dimensions (eg. thin plate splines). \method~allows us to work with multivariate local polynomial interpolators. These are simple objects, and allow us to maintain a sparse representation of our interpolation matrix. In addition, because we work with a discrete approximation to our penalty of choice, we can simultaneously use sparse representations of our discretized penalty and our interpolator. This allows us to easily extend our method and computation to multiple features and thousands of observations.
	
	\section{Multivariate \method}\label{sec:multivar}
	
	Consider the multivariate inputs $\boldsymbol{x}_1,\ldots, \boldsymbol{x}_p \in \mathbb{R}^p$, where $\boldsymbol{x}_j \in [a_j,b_j]$. Similar to before, we assume a generative model of the form
	\[
	y_i = f^* \left(x_{i1},\ldots, x_{ip} \right) + w_i
	\]
	where $f^*$ is an unknown function from a known function class $\mathcal{F}$, and $w_i$ are iid errors with $\operatorname{E}\left[w_i\right] = 0$ and $\operatorname{var}{w_i} = \sigma^2 < \infty$. For multivariate \method, similar tasks have to be completed as univariate \method, namely the selection of the mesh for each considered variable, as well as choosing an interpolation scheme and the Riemann approximation/finite difference order. Let $D\equiv \{D_1,\ldots,D_p \}$, where $D_j$ is the $m_j$-mesh for the $j$th covariate, i.e. 
	\[
	D_j: a_j \le d_{1j} \ldots \le d_{m_j j}=b_j.
	\]
	It is not immediately clear how to interpolate on the multivariate mesh $D$, much less how to approach the approximation $P_D(f_D)$ of $P(f)$. We begin by describing the class we assume for $P(f)$ and our approximation approach. Then we describe the moving local polynomial in the multivariate case. In the supplement, we provide simple matrix representation for bivariate \method. 
	
	\subsection{ Riemann Approximations to Sobolev-like Norms on a Bivariate Mesh:}
	
	We begin by introducing the Sobolev-like form we assume for $P(f)$ in the general case, i.e. $p\ge 2$. Suppose for $p$ covariates, we are interested in $p$ orders of differences given in the multi-index $\boldsymbol{r}=\left( r_1,\ldots, r_p\right)$. Let $\mathcal{D}^{\boldsymbol{r}} f = \frac{\partial^{|\boldsymbol{r}|} f}{\partial x_1^{r_1}\ldots\partial x_p^{r_p}}$ denote the analytic partial derivative, where $|\boldsymbol{r}| = \sum_{i=1}^{p} r_i$. In general, we may be interested in collections of partial derivatives, i.e. $\{\boldsymbol{r}_1,\ldots,\boldsymbol{r}_S\}$. We assume Sobolev-like norms of the following form: 
	\begin{equation}
	P(f)= \|f\|_\ell^\ell = \begin{cases}
	\sum_{s=1}^{S} \|\mathcal{D}^{\boldsymbol{r}_s} f\|_\ell^\ell &\text{ $1\le \ell < +\infty $}\\
	\sup_{s} \|\mathcal{D}^{\boldsymbol{r}_s} f\|_\infty &\text{$\ell=\infty$}
	\end{cases}.
	\end{equation}
	For example, with $p=2$, the collection of first order differences, $\{\boldsymbol{r}_1=(1,1), \boldsymbol{r}_2=(1,0), \boldsymbol{r}_3=(0,1)\}$, specifies the fused lasso bivariate analog \citep{geer2000empirical}.
	
	For our approximation, we introduce some notation. Suppose we have $p$ covariates each with a regular mesh of sizes denoted in the vector $\boldsymbol{m}=(m_1,\ldots,m_p)$. We indicate the bin widths for each of the meshes by $\boldsymbol{\delta}=(\delta_1,\ldots,\delta_p)$, where $\delta_j = d_{i+1,j}-d_{i,j}$ for any $i$. Furthermore, we will assume in this section that 
	\[
	\mathbb{R}^{\boldsymbol{m}}\equiv \mathbb{R}^{m_1}\times \ldots \times \mathbb{R}^{m_p}.
	\] 
	The functional values on the grid are denoted by the $p$-tensor $f_D \in \mathbb{R}^{\boldsymbol{m}}$. Let $(f_D)_{\boldsymbol{i}}=f(d_{\boldsymbol{i}})= f(d_{i_1,\ldots, i_p})$, where $\boldsymbol{i}=(i_1,\ldots,i_p)$. Furthermore, we denote the unit vectors of length $p$ by $\ej$ for $j=1,\ldots,p$, where
	\[
	\ej = (0,\ldots \overbrace{,1,}^{j-th} 0,\ldots,0).
	\]
	
	Recall the univariate normalized first order difference function for an $m$-mesh $\Delta^1_m: \mathbb{R}^m\to \mathbb{R}^{m-1}$ defined in the previous section. We generalize $\Delta^1_m$ (via an extra index) so that we have a normalized first order difference function for the $j$th covariate and any pair of indices such that $\Delta_{\boldsymbol{m},j}^1: \mathbb{R}^{\boldsymbol{m}}\to \mathbb{R}^{\boldsymbol{m}-\ej}$ and  
	\[
	\left[\Delta_{\boldsymbol{m},j}^1f_D \right]_{\boldsymbol{i}} = \frac{f(d_{\boldsymbol{i}+\ej})-f(d_{\boldsymbol{i}})}{\delta_j}.
	\]
	For any pair of indices, we define the $r$th order normalized difference operator $\Delta_{\boldsymbol{m},j}^r: \mathbb{R}^{\boldsymbol{m}}\to \mathbb{R}^{\boldsymbol{m}-r\ej}$ by the recursive formula
	\[
	\left[\Delta_{\boldsymbol{m},j}^rf_D\right]_{\boldsymbol{i}}=\left[\Delta_{\boldsymbol{m}-\ej,j}^{r-1}\left[\Delta_{\boldsymbol{m},j}^1f_D\right]\right]_{\boldsymbol{i}},
	\]
	where $i_j=1,\ldots,m_j-1$ and $i_r=1,\ldots,m_r$ for $r=1,\ldots,p$ ($r\ne j$).
	
	With the generalized first order difference, we approximate $\mathcal{D}^{\boldsymbol{r}}f$ by
	\[
	\Delta^{\boldsymbol{r}}_{\boldsymbol{m}}f_D = \Delta_{\boldsymbol{m}_1,1}^{r_1}\Delta_{\boldsymbol{m}_2,2}^{r_2}\ldots\Delta_{\boldsymbol{m}_p,p}^{r_p}f_D,
	\]
	where $\boldsymbol{m}_p=\boldsymbol{m}$ and $\boldsymbol{m}_{j-1}=\boldsymbol{m}_j-r_j\ej$. Thus, our $\boldsymbol{r}=(r_1,\ldots,r_p)$-order Riemann approximation of $P(f)=\|f\|_\ell^\ell$ using regular meshes for each covariate is given by
	\begin{align}
	P_D(f_D) &= \sum_{\boldsymbol{i}\preceq \boldsymbol{m}-\boldsymbol{k}} |(\delta_1\delta_2\ldots\delta_p)^{1/\ell}[\Delta^{\boldsymbol{r}}_{\boldsymbol{m}}f_D]_{\boldsymbol{i}}|^\ell \\ 
	&= \sum_{\boldsymbol{i}\preceq \boldsymbol{m}-\boldsymbol{k}} |(\delta_1\delta_2\ldots\delta_p)^{1/\ell}[\Delta_{\boldsymbol{m}_1,1}^{r_1}\Delta_{\boldsymbol{m}_2,2}^{r_2}\ldots\Delta_{\boldsymbol{m}_p,p}^{r_p}f_D]_{\boldsymbol{i}}|^\ell,
	\end{align}
	where $\boldsymbol{i} \preceq \boldsymbol{m}-\boldsymbol{r}=\{i_1\le m_1-r_1,\ldots, i_p-r_p\}$. 
	For a collection of partials $\{\boldsymbol{r}_1,\ldots,\boldsymbol{r}_S\}$, we use the following approximation:
	\[
	P_D(f_D)= \sum_{s=1}^{S} \sum_{\boldsymbol{i}\dot{\le} \boldsymbol{m}-\boldsymbol{r}_s} |(\delta_1\delta_2\ldots\delta_p)^{1/\ell}[\Delta^{\boldsymbol{r}_s}_{\boldsymbol{m}}f_D]_{\boldsymbol{i}}|^\ell
	\]
	
	As can be seen, the Riemann approximation extends in a straightforward manner from the univariate case to the multivariate case. In the supplement, we present a matrix notation for both the univariate and bivariate \method~problems.  
	
	\subsection{ Multivariate Interpolation: }
	
	We now describe our approach to multivariate interpolation from a mesh for a sample point $\boldsymbol{x}=(x_{1},\ldots,x_{p})$. The moving local polynomial will be particularly attractive computationally, since we will interpolate $\boldsymbol{x}$ using the minimal number of points needed for a $k$th degree polynomial while keeping $k$ small. First, we update some of the notation previously defined in the section on the univariate moving local polynomial. Recall the previously defined $\tilde{\boldsymbol{d}}$, which denoted in the univariate case the neighborhood of $k+1$ mesh values about a sample point. In parallel, $N^{\boldsymbol{x}}$ denotes a neighborhood of the mesh surrounding $\boldsymbol{x}$. For an order $R$ interpolation, $N^{\boldsymbol{x}}$ contains the $L= {k+p \choose p}$ nearest mesh elements, e.g. $N^{\boldsymbol{x}}=\{\boldsymbol{d}_{1},\ldots,\boldsymbol{d}_{L} \}$. We denote the fitted values for the mesh points used in the interpolation as $\theta_{N^{\boldsymbol{x}}}=\left( \theta_{1},\ldots,\theta_{L}\right)$.
	
	Next, we discuss multivariate polynomial interpolation then introduce the multivariate moving local polynomial. Suppose we want to interpolate $\boldsymbol{x}=(x_1,\ldots,x_p)$ via an $k$th order polynomial: An $k$th order polynomial in general form is:
	\[
	f_k(\boldsymbol{x}) = \beta_{0} + \sum_{j\le p} \beta_{j}x_j + \sum_{j_1 \leq j_2 \leq p} \beta_{j_1,j_2} x_{j_1}x_{j_2} +\ldots + \sum_{j_1\leq \ldots j_k \leq p} \beta_{j_1,\ldots, j_k} x_{j_1}\cdots x_{j_k}.
	\]
	For an $k$th order polynomial in $p$ dimensions we have $T=\left[1+p + ({p \choose 2} + p) + \ldots\right]$ total parameters contained in $\boldsymbol{\beta}=\left(\beta_0,\beta_1,\ldots,\beta_p,\ldots \right)^\top \in \mathbb{R}^T$. Using basis elements as in previous sections, we can write $f_R({\bf x})$ as
	\[
	f_k({\bf x})  = \boldsymbol{\beta}^{\top} \left[\psi_1({\bf x}),\ldots, \psi_T({\bf x})\right]^{\top},
	\]
	where 
	\begin{align*}
	\psi_1({\bf x})&=1, \psi_2({\bf x})=x_1, 	\ldots,\psi_{p+1}({\bf x})=x_p, \\ 
	 \psi_{p+2}({\bf x})&=x_1^2,\psi_{p+3}({\bf x})=x_1 x_2 \ldots, \psi_{2p+2}({\bf x})=x_1 x_p \\ 
	 \psi_{2p+3}({\bf x})&=x_2^2, \ldots  \ldots \psi_{T}({\bf x})=x_p^k.
	\end{align*}
	
	For a point $\boldsymbol{x}$, we will want to form a system of linear equations to interpolate the mesh points contained in $N^{\boldsymbol{x}}$. Using $\{ \boldsymbol{d}_{i} \}_{i=1,\ldots,L} \in N^{\boldsymbol{x}}$, we get a similar system of equations as before:
	\[
	\Psi \boldsymbol{\beta} = \tilde{\theta}
	\]
	with $\tilde{\theta} = [\theta_1, \ldots, \theta_L]$ (our approximation of $\theta_{N^{\boldsymbol{x}}}$), and $\Psi_{ij} = \psi_j({\bf d}_{i})$ ($j=1,\ldots,T$ and $i=1,\ldots,L$). Now when we solve that linear system we get our coefficients for interpolating in that region, i.e. $\Psi^{-1}\tilde{\theta}$. For a new sample point, ${\bf x}_{new}$, in that region, we interpolate with 
	\begin{align*}
	\Omega_{{\bf x}_{new}}(f_D) &= \left(\left[\psi_1({\bf x}_{new}),\ldots, \psi_L({\bf x}_{new})\right] \Psi^{-1}\right)\tilde{\theta} \\
	&= \boldsymbol{a}^\top \tilde{\theta},
	\end{align*}
	where $\boldsymbol{a}$ denote the weights for a linear combination. 
	
	We have described multivariate interpolation over $L$-interpolants for a point of interest ${\bf x}$. We would like a linear operator similar to the univariate case, i.e. $\Omega_{{\bf x}}(f_D)=Of_D$. In the univariate case, the $L=k+1$ nearest mesh points made a neighborhood of $k+1$-consecutive points about an observed data value $x_i$, making the interpolation matrix $O$ banded. However, in this case, the $L$ points will not be consecutive in one direction. For example, in a bivariate scenario with ${\bf m}=(4,4)$, suppose we have 
	\begin{gather*}
	f_D  =
	\begin{pmatrix}
	f(d_{1,1}) & f(d_{1,2})  & f(d_{1,3})  & f(d_{1,4})  \\
	f(d_{2,1}) & f(d_{2,2})  & f(d_{2,3})  & f(d_{2,4})  \\
	f(d_{3,1}) & f(d_{3,2})  & f(d_{3,3})  & f(d_{3,4})  \\
	f(d_{4,1}) & f(d_{4,2})  & f(d_{4,3})  & f(d_{4,4})  \\ 
	\end{pmatrix},
	\end{gather*} 
	The $L$ interpolants to an observation ${\bf x}_i$ could be $\left(f(d_{1,1}),f(d_{1,2}),f(d_{2,2})\right)$. We can define an observation specific interpolation matrix $O_i$ such that 
	\begin{gather*}
	O_i  =
	\begin{pmatrix}
	a_{1,1} & a_{1,2} & 0 & 0  \\
	0 & a_{2,2} & 0 & 0  \\
	0 & 0 & 0 & 0  \\
	0 & 0 & 0 & 0  
	\end{pmatrix},
	\end{gather*} 
	where ${\bf a} = \left(a_{1,1},a_{1,2},a_{2,2}\right)$ are weights determined as previously described. Using $O_i$, we can describe an interpolation using the inner product $\langle\cdot\rangle$:
	\begin{align}
	\Omega_{{\bf x}_i}(f_D) &= \langle O_i, f_D\rangle \\
	&= {\bf tr}(O_i^\top f_D) \\
	&= a_1f_1+a_2f_2+a_3f_6.
	\end{align}  
	Alternatively, we could define $\vec{f_D}$ as the stacking of the rows of $f_D$ into a single column, i.e. 
	\begin{equation}
	\vec{f_D} = \left(f(d_{1,1}),\ldots,f(d_{1,4}),f(d_{2,1}),\ldots,f(d_{3,1}),\ldots,f(d_{4,1}), \ldots, f(d_{4,4})\right)^\top.
	\end{equation}
	In turn, we could define a vector $\vec{o}_i$ as the stacking of the rows of $O_i$, i.e 
	\[
	\vec{o}_i = \left(a_{1,1},a_{1,2},0,0,0,a_{2,2},0,\ldots,0\right)^\top.
	\]
	Using this notation, we arrive at an interpolation matrix $O$, i.e 
	\[
	O = \left(\vec{o}_1,\vec{o}_2,\ldots,\vec{o}_n\right)^\top. 
	\]
	For multivariate data $X=\left(\boldsymbol{x}_1,\ldots,\boldsymbol{x}_p\right)^\top \in \mathbb{R}^{n\times p}$, we define $\Omega_X(f_D)$ as the interpolation of the observed multivariate data using the $p$-tensor $f_D \in \mathbb{R}^{\boldsymbol{m}}$, i.e.
	\begin{equation}\label{mv_terp}
	\Omega_X(f_D) = O\vec{f_D}.
	\end{equation}
	Note that~\eqref{mv_terp} applies for any $p$. When the dimensions of the tensor $f_D$ grow with $m$ and $p$, $\vec{f_D}$ will grow in length (because of the stacking). 
	
	Briefly, we explain some considerations on $k$. Recall that in univariate \method, we require $k\le r$. In multivariate \method, the collections of partials $\{\boldsymbol{r}_1,\ldots,\boldsymbol{r}_S\}$ can designate both isotropic (same order) and anisotropic (mixed order) differences. For example, in a bivariate setting, the collection of first differences contains the isotropic difference, $\boldsymbol{r}_1=(1,1)$, and the anisotropic differences $\boldsymbol{r}_2=(1,0)$ and $\boldsymbol{r}_3=(0,1)$. Here, $k\le 1$, since we are assuming first order smoothness in both predictors. Let $\{\boldsymbol{r'}_1=(1,1),\ldots,\boldsymbol{r'}_{S'}=(S',S')\}$ denote the isotropic differences such that $1\le2\ldots\le S'$ for integer $S'\ge 1$. The rule we follow with multivariate \method~is $k\le S'$.   
	
	\subsection{Writing the Multivariate \method~Objective:}
	
	Suppose we observe response $y_i=f(\boldsymbol{x }_i)+w_i$ for multivariate predictors $\boldsymbol{x}_i\in \mathbb{R}^{p}$ ($i=1,\ldots,n$) with $w_i\sim (0,\sigma^2)$. The $\{\boldsymbol{r}_1,\ldots,\boldsymbol{r}_S\}$-order \method~with $k$th-order interpolation estimates $\tilde{f_D}= \left(f(d_1),\ldots,f(d_m)\right)^\top$ on a regular $m$-mesh are given by 
	\begin{equation}{\label{eq:mvmbs}}
	\min_{f_D \in \mathbb{R}^{\boldsymbol{m}}} \frac{1}{n}\sum_{i=1}^{n}\left(y_i-\langle O_i, f_D\rangle \right)^2+\lambda \sum_{s=1}^{S} \sum_{\boldsymbol{i}\dot{\le} \boldsymbol{m}-\boldsymbol{r}_s} |(\delta_1\delta_2\ldots\delta_p)^{1/\ell}[\Delta^{\boldsymbol{r}_s}_{\boldsymbol{m}}f_D]_{\boldsymbol{i}}|^\ell.
	\end{equation}
	$O_i$ is the $k$th-order interpolation matrix specific to an observation $\boldsymbol{x}_i$ as described in the previous subsection. 
	
	Often, it will be useful to use $O$ and $\vec{f_D}$, the stacked versions of the interpolation matrices $O_1,\ldots,O_n$ and the $p$-tensor $f_D$. We can rewrite the problem in~\eqref{eq:mvmbs} as
	\begin{equation}{\label{eq:mvmbs2}}
	\min_{f_D \in \mathbb{R}^{\boldsymbol{m}}} \left\|y-O \vec{f_D} \right\|_2^2+\lambda \sum_{s=1}^{S} \sum_{\boldsymbol{i}\dot{\le} \boldsymbol{m}-\boldsymbol{r}_s} |(\delta_1\delta_2\ldots\delta_p)^{1/\ell}[\Delta^{\boldsymbol{r}_s}_{\boldsymbol{m}}\vec{f_D}]_{\boldsymbol{i}}|^\ell
	\end{equation} 
	or
	\begin{equation}{\label{eq:bigmvmbs2}}
	\min_{f_D \in \mathbb{R}^{\boldsymbol{m}}} \left\|y-O \vec{f_D} \right\|_2^2+\lambda \|\mathcal{D}\vec{f_D}\|_1,
	\end{equation} 
	for carefully constructed difference operator $\mathcal{D}$. Our fitted values for $X$ are given by $\tilde{f}=O\hat{\vec{f_D}}$.  
	
	\section{Solving \method~Optimization}

	
	In this section, we describe an ADMM solver for the univariate \method~problem when $\ell=1$. Note that in this case the solver will be similar to the standard ADMM algorithm for trend filter given by \citep{ryan2015fast}. For an ADMM as described by \cite{boyd2011distributed}, we begin by rewriting~\eqref{eq:app2} with $\ell=1$ as
	\begin{equation}
	\min_{f_D\in\mathbb{R}^m, \alpha\in \mathbb{R}^{m-k-1}} \|y - Of_D \|_2^2 + \lambda_n \left\|\alpha\right \|_1  \text{ subject to } \alpha = \Delta_{m}^{(r+1)} f_{D}.
	\end{equation}  
	We write the augmented Lagrangian as 
	\[
	L(f_D,\alpha,u) = \|y-Of_D\|_2^2+\lambda\|\alpha\|_1+\frac{\rho}{2} \|\alpha-\Delta_{m}^{(r+1)} f_D+u\|_2^2-\frac{\rho}{2}\|u\|_2^2, 
	\]
	from which we define the following ADMM updates:
	\begin{align}
	f_D &\leftarrow \left(O^\top O + \rho(\Delta_{m}^{(r+1)})^\top \Delta_{m}^{(r+1)}\right)^{-1}\left(O^\top y +\rho (\Delta_{m}^{(r+1)})^\top(\alpha+u)\right), \\
	\alpha & \leftarrow S_{\lambda/\rho}\left(\Delta_{m}^{(r+1)} f_{D}-u\right), \\
	u &\leftarrow u+\alpha-\Delta_{m}^{(r+1)} f_{D}.
	\end{align}
	When using $k$th order moving local polynomials, $O$ becomes banded with bandwidth $k+2$. Meanwhile, $\Delta_{m}^{(r+1)}$ also is banded with bandwidth $r+2$. Since $k\le r$, the $f_D$-update can be implemented with time $O\left(m(r+2)^2+n(r+2)\right)$ and $O\left(m(r+2)^2\right)$ after the first iteration with caching. Updating $\alpha$ with coordinate-wise soft-thresholding ($S_{\lambda/\rho}$) requires time $O(m-k-1)$, while updating $u$ takes $O\left(m(r+2)\right)$ time. Considering $k$ and $r$ as constants, a full iteration of ADMM updates can be done in linear time. 
	
	Solving the multivariate \method~problem given in~\eqref{eq:bigmvmbs2} via ADMM follows easily. In the steps for the univariate ADMM, simply replace  $f_D$ with $\vec{f_D}$ and $\Delta_{m}^{(r+1)}$ with $\mathcal{D}$. An implementation of \method~was programmed using Python (availability). We use this described ADMM solver to run the simulation studies in the following section.        
	
	\section{Simulation Study}
	\vspace{-.05in}
	In this section, we conduct a simulation study showing how the statistical and approximation error decrease as functions of $m$, $n$, $r$, and $k$ for \method~solutions. We generate a response $y_i = f(x_i)+\epsilon_i$, where $f(z)=e^{\pi z}$, $x_i\in U[0,1]$ and $\epsilon_i\sim N(0,1)$ with sample sizes of $n=40,80,120$. For each $(x,y)$-pair, we solve the $\ell_1$ \method~problem using $m=4,5,6,7,8,10,20,30,90$, $r=0,1,2$ and $k=0,1,2$ (via MLP). Note that we maintain $k\le r$ by solving the following pairings: $r=0$ and $k=0$; $r=1$ and $k=0,1$; as well as $r=2$ and $k=0,1,2$.
	
	When tuning $\lambda$, we choose 50 logarithmically spaced values from $10^{-3}$ to $\lambda_{max}$, where 
	\[
	\lambda_{max} = \left\| \left(O\left(\Delta_{m}^{(r+1)}\right)^{-1}\right)^{\top}y\right\|_{\infty}.
	\]
	For each $(m,r,k)$-configuration of \method, we calculate $RMSE = \left(\sum_{j=1}^{500}MSE_j\right)^{\frac{1}{2}}$ where
	\[
	MSE_j = \sum_{i=1}^{n} \left(\tilde{f}(x_i)-f(x_i)\right)^2
	\]
	and $\tilde{f}$ denotes the \method~estimate. 
	 
	\begin{figure} 
		\centering
		\subfloat[ $r=0$ ]{\includegraphics[height=1.5in,width=1.9 in]{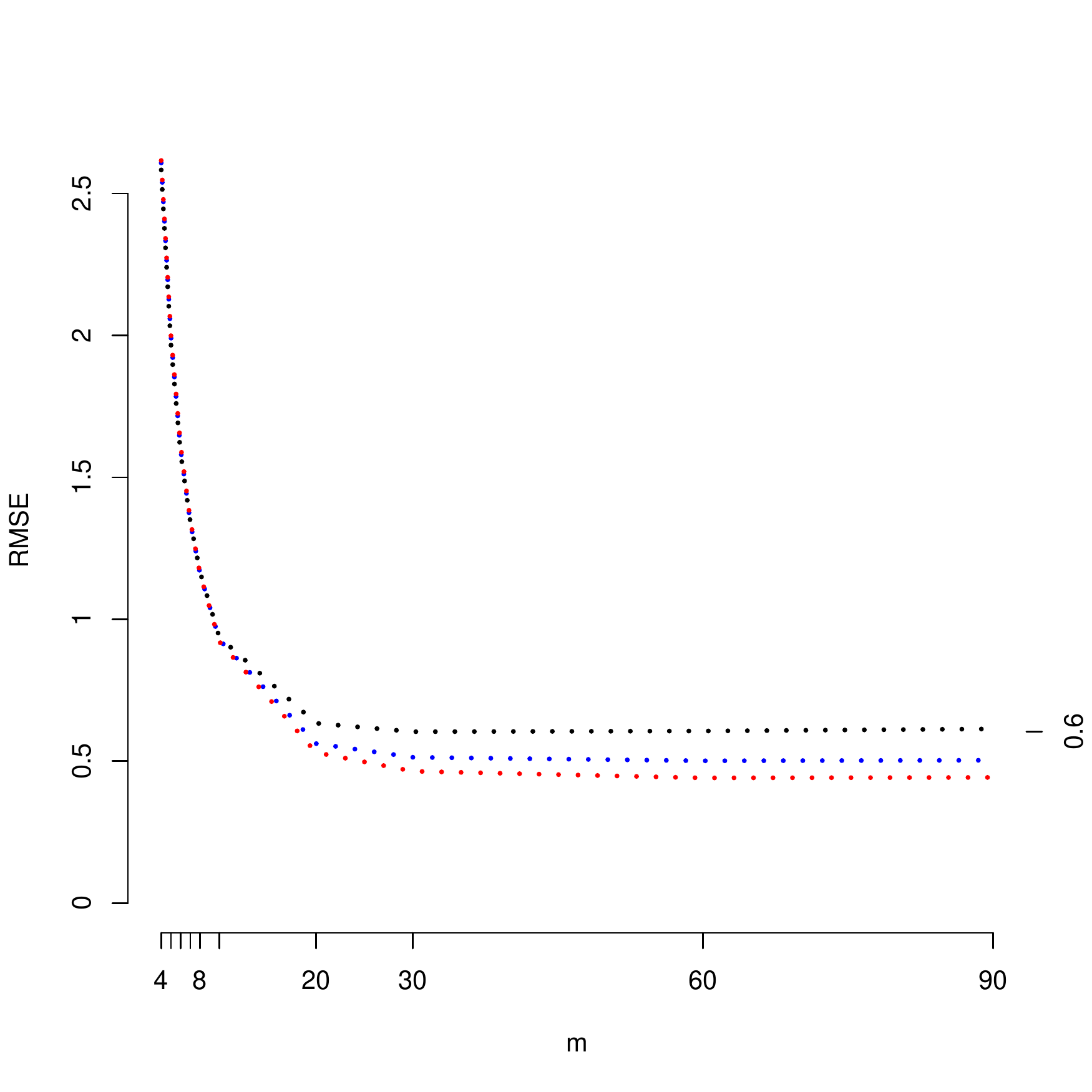}} 
		\subfloat[ $r=1$ ]{\includegraphics[height=1.5in,width=1.9 in]{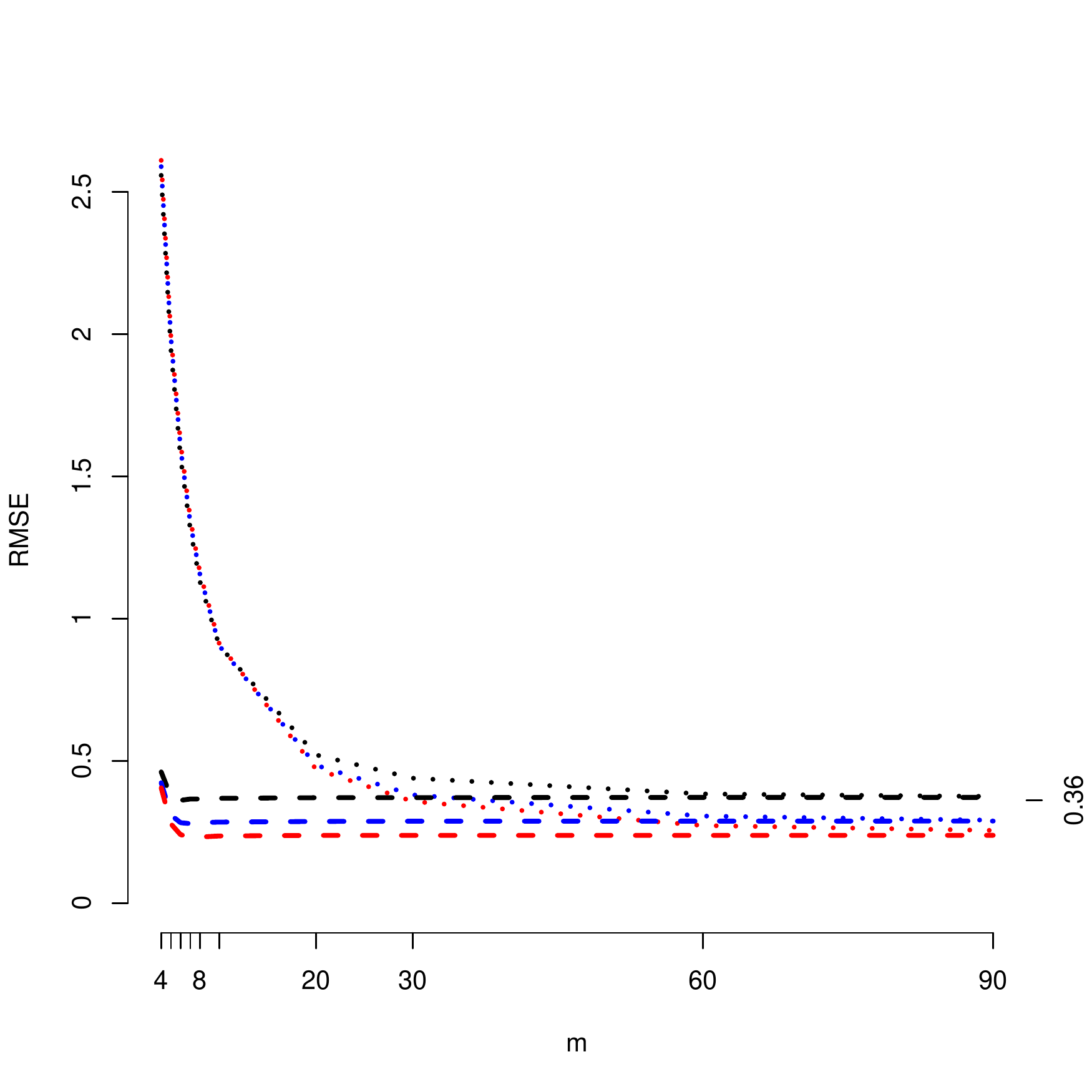}} 
		\subfloat[ $r=2$ ]{\includegraphics[height=1.5in,width=1.9 in]{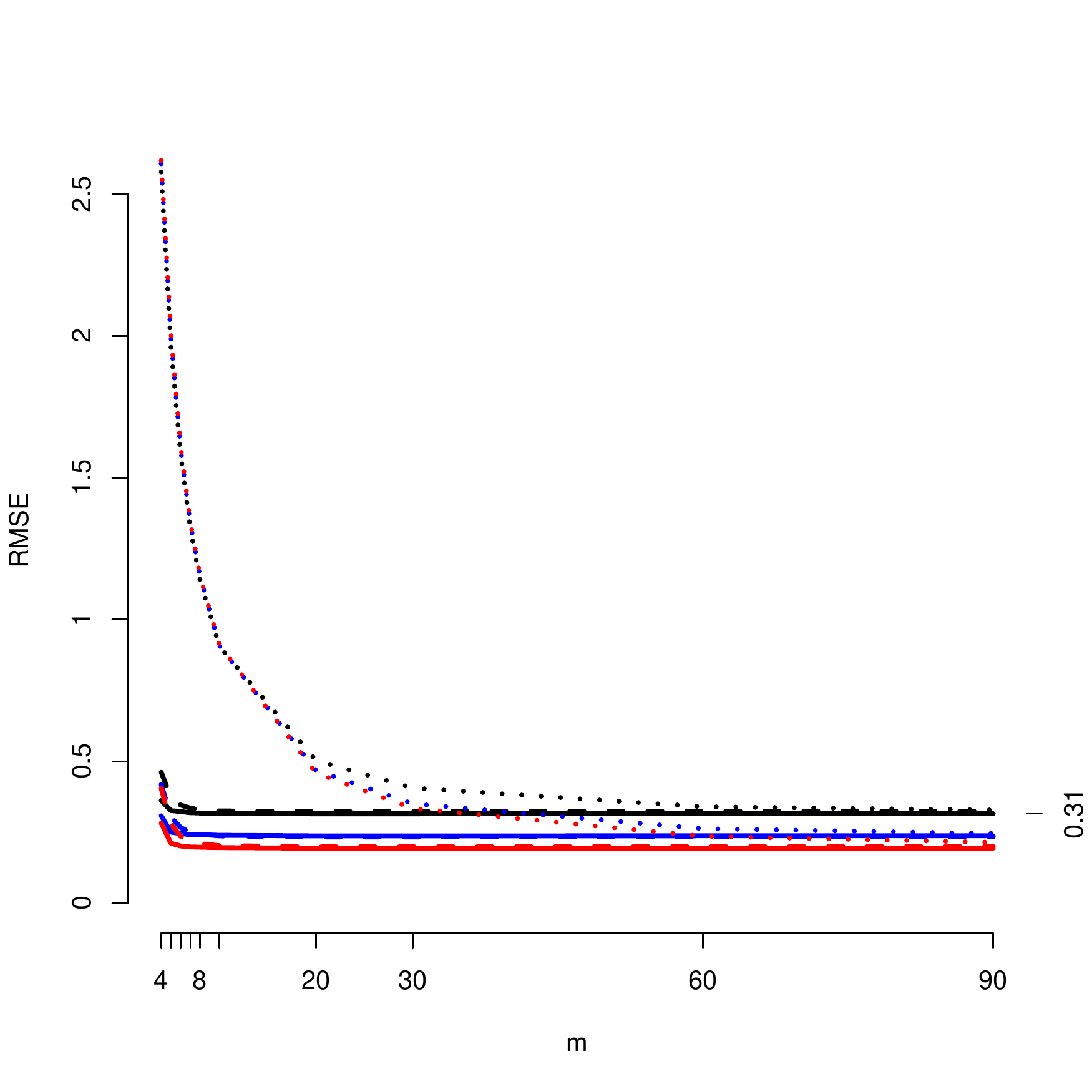}} \\
		\subfloat[ $r=0$ ]{\includegraphics[height=1.5in,width=1.9 in]{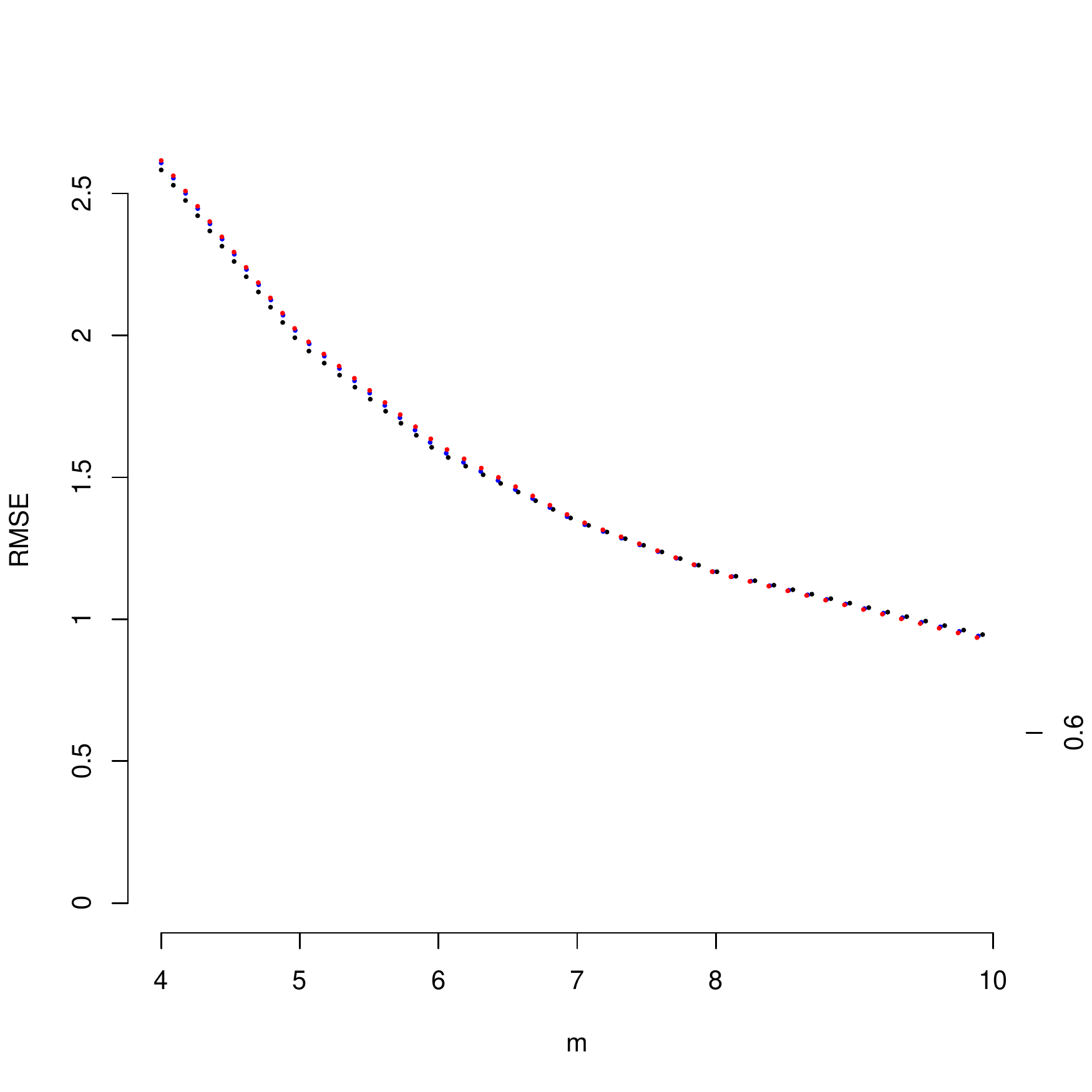}} 
		\subfloat[ $r=1$ ]{\includegraphics[height=1.5in,width=1.9 in]{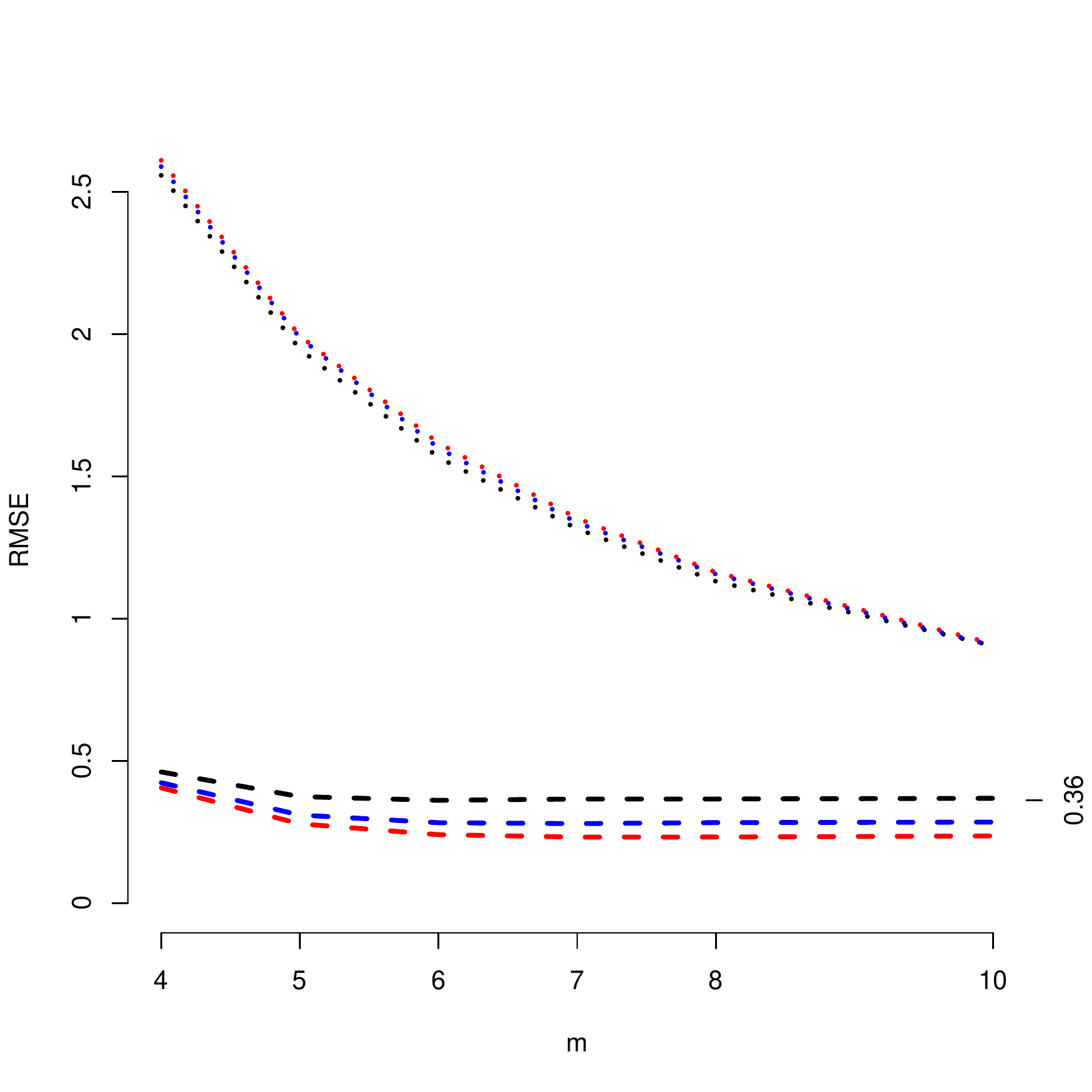}} 
		\subfloat[ $r=2$ ]{\includegraphics[height=1.5in,width=1.9 in]{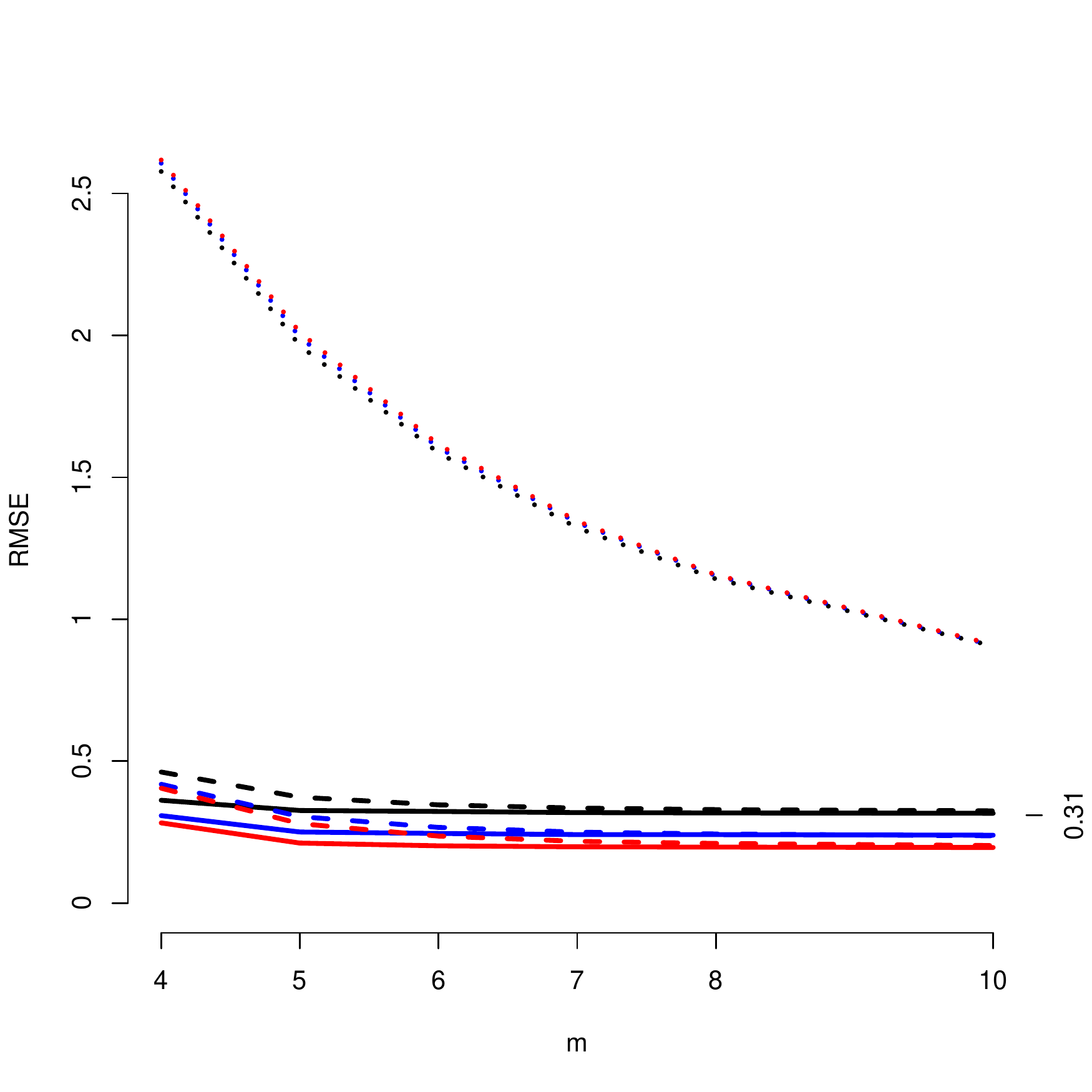}} 
		\caption{Results for 500 simulations over data generated from an exponential function with noise for $n=40, \textcolor{blue}{80}, \textcolor{red}{120}$. \method~models were fit over varying $m$, $r$, and $k$. Line type denotes $k$: 0~(\protect\includegraphics[height = 1em]{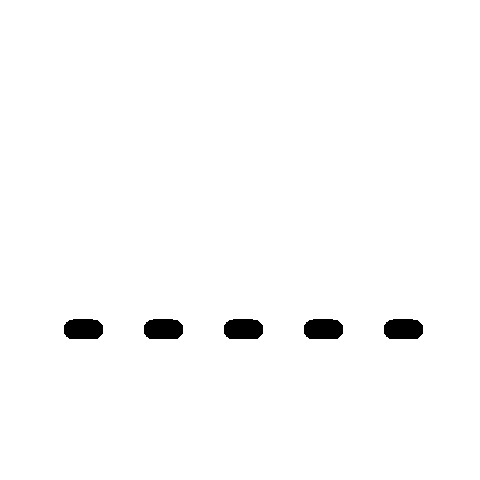}), 1~(\protect\includegraphics[height = 1em]{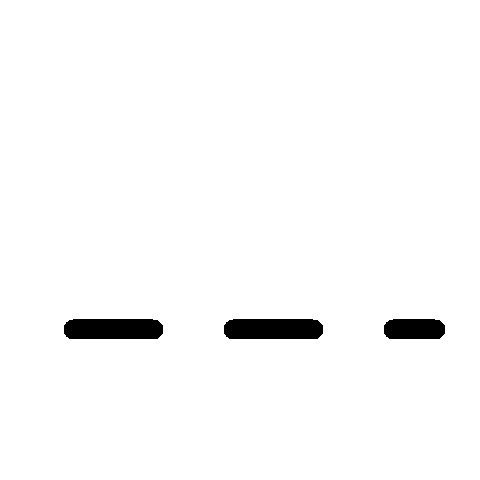}), and 2~(\protect\includegraphics[height = 1em]{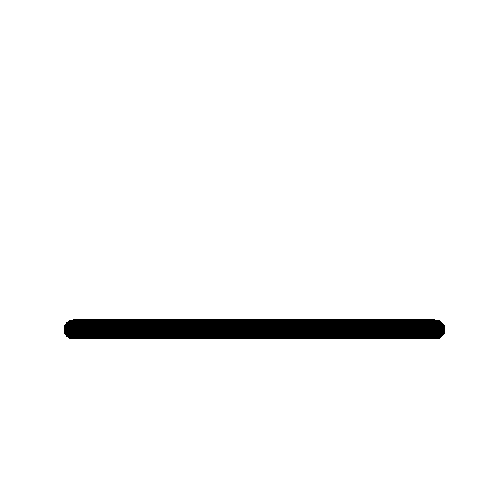}). Top row ranges for all $m$; bottom row ranges for $m\le 10$. }
		\label{fig:r6}
	\end{figure}
	
	As expected, when we hold $r$ and $k$ constant, RMSE tends to decrease towards a limit (specific to $n$) as either $m$ and $N$ increase (Figure 2). We find evidence that the limiting RMSE we approach for $n=40$ and $r=0$ (Figure 2a) is greater than the limiting RMSE for $n=40$ and $r=1$ (Figure 2b). Graphically, it is difficult to see that the limit RMSE for $n=40$ and $r=2$ is in fact the smallest of the three scenarios. This decreasing trend in limiting behavior as $r$ increases holds for other values of $n$, which is not surprising since the underlying exponential function has high order smoothness. Furthermore, we note that as we interpolate at an order close to our strongest assumption of the smoothness, i.e. $k\rightarrow r$, we require smaller $m$ to get close to the limiting RMSE. In Figure 2c when $r=2$, the linear ($k=1$) and quadratic ($k=2$) interpolator are equally as close to the limiting RMSE by $m=6$. Indeed, this fast convergence by the linear interpolator suggests that if we assume sufficient smoothness for a problem, or large enough $r$, we can achieve potentially optimal rates using a linear interpolator with a modest number of knots.     
	
	\subsection{ Simulations for Multivariate \method: }
	
	\begin{figure} 
		\centering
		\subfloat[ True $f$ ]{\includegraphics[height=1.5in,width=1.9 in]{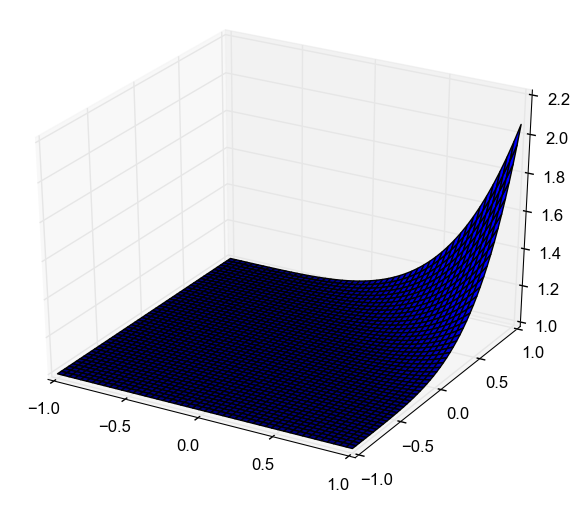}} 
		\subfloat[ $\boldsymbol{m}=(4,4)$ ]{\includegraphics[height=1.5in,width=1.9 in]{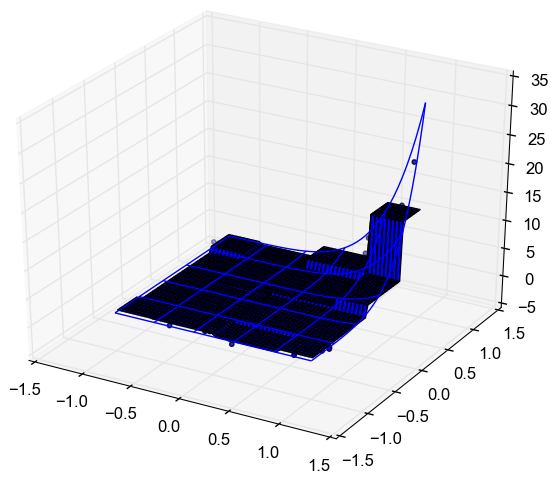}} 
		\subfloat[ $\boldsymbol{m}=(15,15)$ ]{\includegraphics[height=1.5in,width=1.9 in]{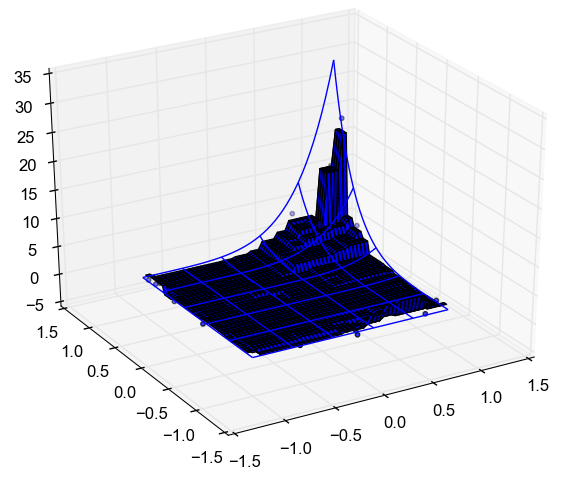}} 
		\caption{We observe $n=100$ noisy observations (transparent) of a bivariate exponential function shown in (a) and wireframe. We draw \method~fits (blue) using $\{\boldsymbol{r}_1=(1,1), \boldsymbol{r}_2=(1,0), \boldsymbol{r}_3=(0,1)\}$ and $k=0$. }
		\label{fig:mv_fl}
	\end{figure}
	
	In the univariate case, we see that \method~resembles many other univariate methods and so we spent time showing how \method~competes. However, in the multivariate setting, $\methodsym$ approximate problems previously thought to be intractable. Let us consider solving over the functions of bounded variation in $\mathbb{R}^2$, i.e 
	\[
	\mathcal{F}=\{f: \mathbb{R}^2\rightarrow \mathbb{R}, P(f)\le M\}
	\]
	, where
	\[
	P(f) = \int \left|\frac{\partial}{\partial x_1}f(x_1,x_2)\right| dx_1 +\int \left|\frac{\partial}{\partial x_2}f(x_1,x_2)\right| dx_2 + \int \left|\frac{\partial^2}{\partial x_1 \partial x_2}f(x_1,x_2)\right| dx_1 dx_2.
	\]
	\citep{geer2000empirical} derived uniform convergence rates for least squares estimates, $\hat{f}^{LS}$, of $f\in \mathcal{F}$. Using \method~with $0$th-order interpolation, we can approximate this bivariate bounded variation problem. In Figure~\ref{fig:mv_fl}, we show \method~fits approximating the bivariate fused lasso for an exponential function. \method~aims to capture the complexity of the exponential curve using $m=15$, i.e. a 15 by 15 mesh on the predictor space. In Figure~\ref{fig:sims_mv_fl}, we show results for a simulation study using much larger sample sizes and a wide range of $m$. For this smooth exponential function, using as small as $m=40$ begins to produce near-optimal fits for large sample sizes.          

	\begin{figure} 
		\centering
		\includegraphics[height=2in,width=2.5 in]{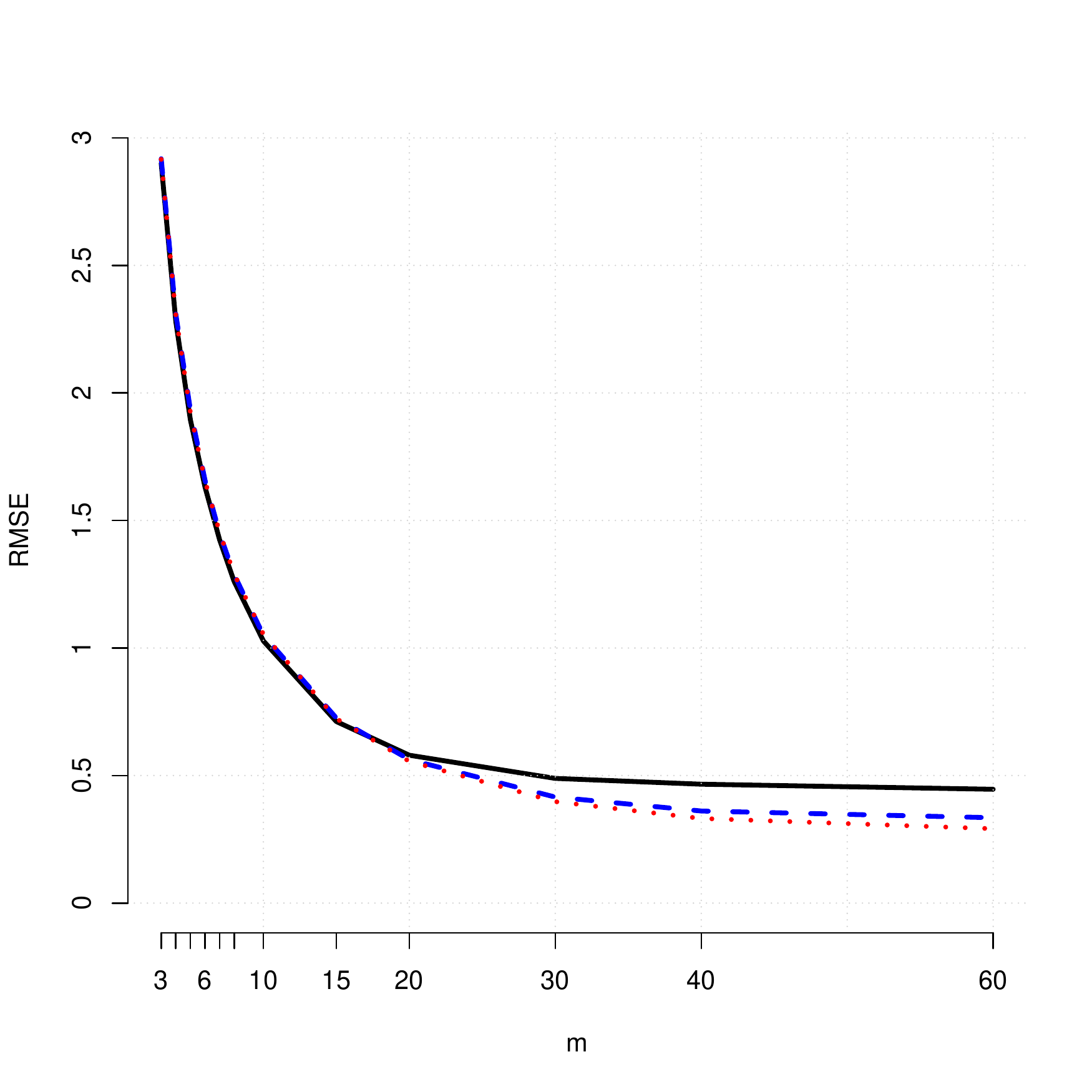} 
		\caption{We simulated $n=1,000$ (\protect\includegraphics[height = 1em]{./Figures_v2/solid.png}), $\textcolor{blue}{5,000}$ (\protect\includegraphics[height = 1em]{./Figures_v2/dashed.png}), $\textcolor{red}{10,000}$ (\protect\includegraphics[height = 1em]{./Figures_v2/dotted.png}) noisy observations of a bivariate exponential function shown in Figure~\ref{fig:mv_fl}. For each of 500 simulations, we approximated the bivariate fused lasso using \method~over a range of $m$. }
		\label{fig:sims_mv_fl}
	\end{figure}
	
	\section{\method~Theory}
	
	In this section, we discuss theoretical properties of the \method~estimate. We introduce some notation while making clear the objectives we are comparing. The \method~estimate of $f^*$ on a mesh $D$ is given by $\methodsym$ where $\tilde{f_D}$ minimizes the following problem:
	\begin{equation}
	\tilde{f_D} = \operatorname{argmin}_{f_D\in\mathbb{R}^m}L_D(f_D),
	\end{equation}
	where
	\[
	L_D\left(f_D\right) = \frac{1}{n}\sum_{i=1}^n\left(y_i - \Oxi \right)^2 + \lambda_n P_D\left(f_D\right).
	\]
	We would like \method~to be comparable to the solution of a total variation problem, $\hat{f}$, given by
	\begin{equation}
	\hat{f} = \operatorname{argmin}_{f\in\mathcal{F}}L\left(f\right)
	\end{equation}
	where 
	\[
	L(f)= \frac{1}{n}\sum_{i=1}^n\left(y_i - f\left(x_{i}\right)\right)^2 + \lambda_n P\left(f\right).
	\]
	\method~is sub-optimal in the sense that
	\[
	L\left(\hat{f}\right) \le L\left(\methodsym\right). 
	\] 
	However, under certain conditions, \method~can be as optimal as penalized regression exact solutions, as we show next. We still assume $P(\cdot)$ takes form of a Sobolev penalty. For this section, we introduce a gridding function: for any function $g$, we define a gridding function $D: \mathcal{F} \rightarrow \mathbb{R}^m$ such that for $g\in \mathcal{F}$, $D(g) = \left(g(d_1),\ldots,g(d_m)\right)^\top$. 
	
	\begin{lemma}
		\emph{(Sub-Optimality Inequality)}
		\label{lemma_subopt}
		For all $f \in \mathcal{F}$, suppose there exist $\delta_m$ and $\epsilon_m$ such that
		\begin{equation}\label{riemannbound}
		\sup_{ f \in \mathcal{F}}\left|P_D\left(D\left(f\right)\right)-P(f)\right| \le \epsilon_m
		\end{equation}
		and
		\begin{equation}\label{interpbound}
		\sup_{ f \in \mathcal{F}}|\Omega_x\left(D\left(f\right)\right)-f(x)| \le \delta_m.
		\end{equation}
		
		With $\lambda_n>0$, we have
		\begin{gather}\label{ineq:subopt}
		L\left(\methodsym\right) \le \min \begin{cases}
		3L(\hat{f}) + O_P\left(\delta_m^2\vee \epsilon_m\lambda_n\right) \\
		L(\hat{f}) + O_P\left(C\delta_m \vee \epsilon_m\lambda_n\right)
		\end{cases}.
		\end{gather}
	\end{lemma}
	\begin{proof}
		See supplement.
	\end{proof}
	
	Thus, Lemma~\ref{lemma_subopt} shows \method~can perform as well as the exact penalized regression solution depending on the penalty approximation error~\eqref{riemannbound} and interpolation error~\eqref{interpbound}. To prove a theorem about a rate of convergence for \method, we utilize Lemma~\ref{lemma_subopt} and entropy theory. Let $H\left(\delta, \mathcal{F}, Q_n\right)=\operatorname{log} N\left(\delta, \mathcal{F}, Q_n \right)$ denote the $\delta$-entropy of $\mathcal{F}$ for the $L_2(Q_n)$-metric, where $N\left(\delta, \mathcal{F}, Q_n \right)$ is the $\delta$-covering number. With $0<\alpha<2$, we suppose that 
	\begin{gather}\label{eq:EB}
	H\left(\delta,\{f\in \mathcal{F}: P(f)\le 1 \},Q_n\right) \le c\delta^{-\alpha},
	\end{gather}
	e.g. this holds for many functional classes, such as functions with bounded total variation or in Sobolev spaces. Using these entropy conditions, we show in the following theorem that \method~achieves a similar rate of convergence as the exact solution of~\eqref{eq:pen}: 
	
	\begin{theorem}
		\emph{(Rate of Convergence)}
		\label{lem:roc1}
		Assume standard entropy conditions on $\alpha$ given in the supplement. Let $P(f^*)>0$ and $\lambda_n=O_p\left(n^{-\frac{2}{2+\alpha}}\right)$. If 
		\begin{gather}
		L\left( \methodsym \right) \le L(\hat{f})+\Gamma_{n,m},
		\end{gather}
		then we have 
		\begin{gather}\label{eq:result1}
		\left\|\methodsym-f^*\right\|_n^2 =O_p\left(\lambda_n + 6\Gamma_{n,m} \right).
		\end{gather}
	\end{theorem}
	\begin{proof}
		See supplement. 
	\end{proof}
	
	Based on Lemma~\ref{lemma_subopt}, we know that $\Gamma_{n,m}=O_P(\delta_m \vee \epsilon_m\lambda_n)$. Since $\lambda_n\rightarrow 0$ and $\epsilon_m\rightarrow 0$, we are not concerned about $\epsilon_m\lambda_n$. Hence, in this case, $\Gamma_{n,m}$ is the excess error due to interpolation. Ideally, $\lambda_n\ge \delta_m$ so that
	\[
	\left\|\methodsym-f^*\right\|_n^2 =O_p\left(\lambda_n\right),
	\] 
	i.e. a \method~estimate attains the optimal rate. 
	
	\subsection{ \method~Convergence Rates using Polynomial Interpolators: }
	
	In the simulations of this paper, we used a polynomial interpolator, i.e. MLP, to fit the data. It is not difficult to derive point-wise rates for the interpolation error from an interpolating polynomial of $k$th degree on a regular $m$-mesh (see supplement): 
	\begin{equation}
	|\Omega_{x_i}\left(D(f)\right)-f(x_i)| = O(m^{-(k+1)}).
	\end{equation}
	Ideally, the cost of fitting on a mesh is dominated by the cost of the penalized regression problem, $\lambda_n$, i.e.
	\begin{align*}
	n^{-\frac{2}{2+\alpha}} & >  m^{-(k+1)}.
	\end{align*}
	Fortunately, modestly grown $m$ achieves this property. Conservatively, if $m>n^{\frac{1}{k+1}}$, then $\lambda_n>n^{-1}>m^{-(k+1)}$. In Figure~\ref{fig:r6}, at a fixed $k$ regardless of $r$, we saw no further changes in RMSE by $m>n^{\frac{1}{k+1}}$, which may indicate that the \method~estimates have achieved the same limiting RMSE as exact solutions.
	
	\section{Conclusion}
	
	It can be intractable to solve the exact problem given by~\eqref{eq:pen}. We have introduced an inexact problem, the \method~objective given by \eqref{eq:method}, whose calculable solutions efficiently approximate~\eqref{eq:pen}. \Method approximates~\eqref{eq:pen} via discretization: we interpolate $f_D$, i.e. $\Omega_{x_i}(f_D)$, to fit $f(x_i)$; and we approximate $P(f)$ with $P_D(f_D)$. Using simple interpolators and Riemann sums of differences, \method~allows us to solve previously intractable problems especially in the multivariate setting. In future work, we hope to study the behaviour of \method~estimates in the high dimensional setting as well as prove uniform rates for the interpolation error or the Riemann approximation error.  
	
	\bigskip
	\begin{center}
		{\large\bf SUPPLEMENTAL MATERIALS}
	\end{center}

	\section{Notation and Theory}
	
	\subsection{Matrix Notation of Univariate/Bivariate \method}
	
	We denote a first order difference matrix by
	\begin{gather}
		\Delta^{(1)}_n  =
		\begin{pmatrix}
			-1 & 1 & 0 & \ldots  & 0 & 0 \\
			0 & -1 & 1 & \ldots & 0 & 0 \\
			\vdots &  &  &  & & \\
			0 & 0 & 0 & \ldots & -1 & 1 	
		\end{pmatrix} \in \mathbb{R}^{(n-1)\times n}.
	\end{gather} 
	The $r$th order difference matrix is defined recursively as follows:
	\begin{gather}
		\Delta^{(r)}_{n} = \Delta^{(r-1)}_{n-1} \cdot \Delta^{(1)}_n =\Delta^{(1)}_{n-r+1}\cdot \Delta^{(1)}_{n-r} \cdot \ldots \cdot \Delta^{(1)}_n \in \mathbb{R}^{(n-r)\times n}.
	\end{gather}
	We define an averaging operator matrix by 
	
	\begin{gather}
		A^k_{n} = \frac{1}{r+1} 
		\begin{pmatrix}
			\bovermat{(r+1)-many}{1 & 1 & \ldots &  1} & 0 & \ldots & 0 \\
			0 & 1 & \ldots & 1 & 1 & \ldots & 0 \\
			\vdots &  &  &  & &  & \\
			0 & 0 & 0 & \ldots & 1 & 1 & 1 	
		\end{pmatrix} \in \mathbb{R}^{(n-r)\times n}.
	\end{gather} 
	$A_m^0$ gives the identity matrix.
	
	Suppose $p=1$ and we choose an $m$-mesh $D$ (not necessarily regular). We use finite-differences/Riemann sums to approximate the Sobolev norm denoted by $P(f)$:
	\begin{gather}\label{meshy_mat}
		P_D(f_D) = \left\|\bar{\Delta}_{m}^{(r)} f_D\right\|_\ell^\ell
	\end{gather}
	where $f_D = \left[f\left(d_1\right),\ldots, f\left(d_m\right)\right]^{\top}$ and $\bar{\Delta}_{m,D}^{(r)}$ is an $r$th order normalized difference over $D$:
	\begin{gather}\label{normdiffmat}
		\bar{\Delta}_{m}^{(r)} = \left(\tilde{\Delta}_{m}^{(r)} \right)^{1/\ell} \left((\tilde{\Delta}_{m}^{(r)})^{-1}\Delta_{m-r}^{(1)}\right) \left((\tilde{\Delta}_{m}^{(r-1)})^{-1}\Delta_{m-r-1}^{(1)}\right)	\ldots \left((\tilde{\Delta}_{m}^{(1)})^{-1}\Delta_{m-1}^{(1)}\right)(\tilde{\Delta}_{m}^{(0)})^{-1}\Delta_{m}^{(1)},
	\end{gather}
	with
	\[
	\tilde{\Delta}_{m}^{(r)}=\operatorname{diag}\left(\Delta_{m-r}^{(1)} A_m^r D\right)\in \mathbb{R}^{(m-r-1)\times (m-r-1)}.
	\]
	$A_m^rD$ produces the averages of the $r+1$ adjacent values of the mesh. $\Delta_{m-r}^{(1)}A_m^rD$ gives the $r$th order differences of the $r$ adjacent mesh point averages.
	
	Through this generalized matrix formulation, we can formulate the approximating norms with some algebra. For example, with $r=1$ and $\ell=2$, we have
	\[
	P_D(f_D) = \sum_{i=1}^{m-2} \frac{\left( \frac{f(d_{i+2})-f(d_{i+1})}{d_{i+2}-d_{i+1}}-\frac{f(d_{i+1})-f(d_{i})}{d_{i+1}-d_{i}}\right)^2}{\frac{d_{i+2}-d_i}{2}}.
	\]  
	
	On a regular $m$-mesh $D$ with mesh widths $\delta$, \eqref{normdiffmat} reduces nicely. For integers $r\ge 0$ and $\ell>0$, our Riemann approximation to $P(f)$ takes the following form:
	\begin{gather}
		P_D(f_D) = \left\|\delta^{\frac{1}{\ell}-r}\Delta_{m}^{(r)} f_D\right\|_\ell^\ell.
	\end{gather}
	
	In the bivariate case with regular meshes chosen for each covariate, we arrive at simple expressions of the Riemann approximation. Let $D=(D_1,D_2)$ denote regular meshes for covariates $x_1$ and $x_2$, respectively, i,e. $\boldsymbol{m}=(m_1,m_2)$. Define $\theta\in \mathbb{R}^{\boldsymbol{m}}$ such that 
	\begin{gather}
		\theta  =
		\begin{pmatrix}
			f(d_{1,1},d_{2,1}) & f(d_{1,2},d_{2,1}) & \ldots  & f(d_{1,m_1},d_{2,1}) \\
			f(d_{1,1},d_{2,2}) & f(d_{1,2},d_{2,2}) & \ldots  & f(d_{1,m_1},d_{2,2}) \\
			\vdots &  &  & \\
			f(d_{1,1},d_{2,m_2}) & f(d_{1,2},d_{2,m_2}) & \ldots  & f(d_{1,m_1},d_{2,m_2})	
		\end{pmatrix}.
	\end{gather} 
	Furthermore, let $\boldsymbol{r}=(r_1,r_2)$ denote the partials we seek to estimate. Similar to previous notation, but with sub-indices for the covariates, the $r_j$th-order differences for the $j$th variable can be calculated through $\Delta^{(r)}_{m_j}\theta^{[j]}$, where $\theta^{[1]}=\theta$ and $\theta^{[2]}=\theta^\top$. In this bivariate case with pure partials, i.e. taking differences only for one covariate or isotropic differences, with mesh widths denoted by $\boldsymbol{\delta}=(\delta_1,\delta_2)$,  
	\[
	P_D\left(f_D\right) = \left\|\delta_j^{\frac{1}{\ell}-r_j}\delta_{j'}^{\frac{1}{\ell}}\Delta^{(r_j)}_{m_j}\theta^{[j]}\right\|_\ell^\ell,
	\] 
	where $j=1,2$ and $j'=1,2$ ($j \ne j'$). 
	
	For mixed partials or anisotropic derivatives with $\boldsymbol{r}=(r_1,r_2)$, we can calculate the approximating differences using 
	\[
	\Delta_{\boldsymbol{m}}^{(\boldsymbol{r})}\theta=\Delta_{m_1}^{(r_1)}\theta (\Delta_{m_2}^{r_2})^\top
	\] 
	or 
	\[
	\Delta_{\boldsymbol{m}}^{(\boldsymbol{r})}\theta =\Delta_{m_2}^{(r_2)}\theta^\top (\Delta_{m_1}^{(r_1)})^\top.
	\] 
	Thus, we can estimate $P(f)$ with 
	\[
	P_D\left(f_D\right) = \left\|\delta_1^{\frac{1}{\ell}-r_1}\delta_{2}^{\frac{1}{\ell}-r_2}\Delta_{\boldsymbol{m}}^{(\boldsymbol{r})}\theta\right\|_\ell^\ell.
	\]
	In the general bivariate case, for $\{\boldsymbol{r}_1,\ldots,\boldsymbol{r}_S\}$, 
	\[
	P_D\left(f_D\right)= \sum_{s=1}^{S}\left\|\delta_1^{\frac{1}{\ell}-r_{1,s}}\delta_{2}^{\frac{1}{\ell}-r_{2,s}}\Delta_{\boldsymbol{m}}^{(\boldsymbol{r}_s)}\theta\right\|_\ell^\ell
	\]
	
	\subsection{Rising Polynomial Basis}
	
	In this section, we define the $K$th order rising polynomial basis. Suppose we define a mesh $D=\left(d_1,\ldots,d_m\right)$. For some integer $Q$, set $m=1+QK$. Let $T=\left(d_{1+K},d_{1+2K},\ldots,d_{1+(Q-1)K}, d_{m}\right)=\left(t_1,\ldots,t_{Q}\right)$. The $K$th order rising polynomial has basis elements given by
	\[
	\psi_1(x)=1, \psi_2(x)=x,\ldots, \psi_{K+1}(x)=x^K, \mbox{ and } \psi_{q,k}(x)=(x-t_{q+1})^k_+-(x-t_{q})^k_+,
	\]
	where $q=1,\ldots,Q$ and $k=1,\ldots,K$.  
	
	\subsection{ Proofs of Theoretical Results }
	
	\subsubsection{Sub-Optimality Lemma}
	
	First, we introduce notation for the estimators compared in this section. The exact solution of PR is given by:
	\begin{equation}
		\hat{f} = \operatorname{argmin}_{f\in\mathcal{F}}L\left(f\right)
	\end{equation}
	where 
	\[
	L(f)= \frac{1}{n}\sum_{i=1}^n\left(y_i - f\left(x_{i}\right)\right)^2 + \lambda_n P\left(f\right).
	\]
	We refer to $L(f)$ as the functional loss of $f$. 
	
	On a mesh $D$, we solve our approximation to penalized regression:  
	\begin{equation}
		\tilde{f_D} = \operatorname{argmin}_{f\in\mathcal{F}}L_D(f_D),
	\end{equation}
	where
	\[
	L_D\left(f_D\right) = \frac{1}{n}\sum_{i=1}^n\left(y_i - \Oxi \right)^2 + \lambda_n P_D\left(f\right).
	\]
	The functional estimate, i.e. \method, is given by $\meshy$. 
	
	At best, \method, $\meshy$, approximates the solution of a penalized regression problem, $\hat{f}$. \method is sub-optimal in the sense that
	\[
	L\left(\hat{f}\right) \le L\left(\meshy\right). 
	\] 
	However, under certain conditions, it can be as optimal as penalized regression, as we show next. Let us assume $P(\cdot)$ takes form of a Sobolev penalty. 
	
	\begin{lemma}
		\emph{(Sub-Optimality Inequality)}
		\label{ineq:suboptshort}
		For all $f \in \mathcal{F}$, suppose there exist $\delta_m$ and $\epsilon_m$ such that
		\begin{equation}
			\sup_{ f \in \mathcal{F}}\left|P_D\left(D\left(f\right)\right)-P(f)\right| \le \epsilon_m
		\end{equation}
		and
		\begin{equation}
			\sup_{ f \in \mathcal{F}}|\Omega_x\left(D\left(f\right)\right)-f(x)| \le \delta_m.
		\end{equation}
		
		With $\lambda_n>0$, we have
		\begin{gather}\label{ineq:subopt2}
			L\left(\meshy\right) \le \min \begin{cases}
				3L(\hat{f}) + O_P\left(\delta_m^2\vee \epsilon_m\lambda_n\right) \\
				L(\hat{f}) + O_P\left(C\delta_m \vee \epsilon_m\lambda_n\right)
			\end{cases}.
		\end{gather}
	\end{lemma}
	
	\begin{proof}
		We know $L_D\left(\tilde{f}_D\right)\le L_D\left(D\left(\hat{f}\right)\right)$. After some algebra, we find
		\begin{align}
			L_D\left(\tilde{f}_D\right) & = L\left(\meshy\right) + \lambda_n\left(P_D\left(\tilde{f}_D\right) -P\left(\meshy\right)\right).
		\end{align}
		With similar algebra and by applying Cauchy-Schwarz Inequality, we show
		\begin{align}
			L_D\left(D\left(\hat{f}\right)\right) &\le L\left(\hat{f}\right) + \delta_m^2 +2\|y-\hat{f}\|_n\|\Omega(D(\hat{f}))-\hat{f}\|_n+ \epsilon_m\lambda_n
		\end{align}
		From here, we have two bounds for $L_D\left(D\left(\hat{f}\right)\right)$. We know $\|y-\hat{f}\|_n^2 \le L(\hat{f})$, since $P(\hat{f})>0$ and $\lambda_n\ge 0$. Hence, we can deduce
		\begin{align}
			L_D\left(D\left(\hat{f}\right)\right) & \le L\left(\hat{f}\right) + \delta_m^2 +2\left(\|y-\hat{f}\|_n^2 \vee \|\Omega(D(\hat{f}))-\hat{f}\|_n^2\right)+ \epsilon_m\lambda_n \\
			& \le L\left(\hat{f}\right) + \delta_m^2 +2\left(L(\hat{f}) \vee \delta_m^2\right)+ \epsilon_m\lambda_n \\
			& \le 3L\left(\hat{f}\right) + O_P\left(\delta_m^2 \vee \epsilon_m\lambda_n\right). 
		\end{align}
		However, by using the fact that $L\left(\hat{f}\right)\le L\left(f^*\right)$, it is not difficult to show 
		\[
		\|y-\hat{f}\|_n \le \|w\|_n +O_P(\lambda_n^{1/2})<C.
		\]
		Hence, another bound follows that will be useful: 
		\begin{align}
			L_D\left(D\left(\hat{f}\right)\right) & \le L\left(\hat{f}\right) + \delta_m^2 +2C\delta_m + \epsilon_m\lambda_n \\
			& \le L\left(\hat{f}\right) + O_P\left(C\delta_m\vee \epsilon_m\lambda_n \right).
		\end{align}
		\qed 
	\end{proof} 
	
	\subsubsection{ Rate of Convergence for $\meshy$ }
	
	Let $H\left(\delta, \mathcal{F}, Q_n\right)=\operatorname{log} N\left(\delta, \mathcal{F}, Q_n \right)$ denote the $\delta$-entropy of $\mathcal{F}$ for the $L_2(Q_n)$-metric, where $N\left(\delta, \mathcal{F}, Q_n \right)$ is the $\delta$-covering number and $Q_n$ denotes the empirical measure. Let us suppose that $\mathcal{F}$ is a cone, and that 
	\begin{gather}
		H\left(\delta,\{f\in \mathcal{F}: P(f)\le 1 \},Q_n\right) \le c_1\delta^{-\alpha},
	\end{gather}
	for all $\delta>0$ and some constants $c_1>0$ and $0<\alpha<2$. Let $v>\frac{2\alpha}{2+\alpha}$. The same entropy bound holds for the normalized functions when $P(f)+P(f^*)>0$:
	\begin{gather}\label{eq:EB2}
		H\left(\delta,\left\{ \frac{f-f^*}{P(f)+P(f^*)}: f\in \mathcal{F}, P(f)+P(f^*)> 0 \right\},Q_n\right) \le c_2\delta^{-\alpha}.
	\end{gather}
	Furthermore, we assume the errors have sub-Gaussian tails:
	\begin{gather}\label{eq:subgauss}
		\sup_n \max_{i=1,\ldots,n} K^2\left(\mathbb{E}e^{|\epsilon_i|^2/K^2}-1\right) \le \sigma^2.
	\end{gather}
	By Lemma 8.4 in \cite{geer2000empirical}, with $P(f^*)>0$,
	\begin{gather}\label{eq:supincrem}
		\sup_{f\in \mathcal{F}} \frac{|(w,f-f^*)_n|}{\|f-f^*\|_n^{1-\alpha/2}(P(f)+P(f^*))^{\frac{\alpha}{2}}} = O_P(n^{-1/2}). 
	\end{gather}
	
	In Theorem 10.1 of \cite{geer2000empirical}, an optimal rate of convergence is established assuming \eqref{eq:EB2} and \eqref{eq:subgauss}. Since the sub-optimality of our estimator has been quantified in Lemma 1.1, we need only modify Theorem 10.1 for a rate of convergence.   
	
	\begin{lemma}
		\emph{(Rate of Convergence)}
		Let $P(f^*)>0$ and 
		\begin{gather}\label{rate:lam}
			\lambda_n=O_p\left(n^{-\frac{2}{2+\alpha}}\right)
		\end{gather}
		for $0<\alpha<2$.
		If 
		\begin{gather}
			L\left( \meshy \right) \le L\left(\hat{f}\right)+\Gamma_{n,m},
		\end{gather}
		then we have 
		\begin{gather}
			\left\|\meshy-f^*\right\|_n^2 =O_p\left(\lambda_n + 6\Gamma_{n,m} \right).
		\end{gather}
	\end{lemma}
	
	\begin{proof}
		Rewriting $L\left( \meshy\right) \le L\left(\hat{f}\right) +\Gamma_{n,m} \le L(f^*)+\Gamma_{n,m}$, we get a basic inequality: 
		\begin{align}
			\left\|\meshy-f^*\right\|_n^2 + \lambda_nP\left(\meshy\right) & \le 2\left(w,\meshy-f^*\right)+\lambda_n P\left(f^*\right) +2\Gamma_{n,m} \\
			& \le 3\max\left\{2\left(w,\meshy-f^*\right)_n, \lambda_n P(f^*), 2\Gamma_{n,m}\right\}.
		\end{align} 
		When $\Gamma_{n,m}$ is the maximum, then \eqref{eq:result1} follows:
		\begin{gather*}
			\left\|\meshy-f^*\right\|_n^2 \le O_P\left(6\Gamma_{n,m}\right). 
		\end{gather*} 
		Otherwise, using similar techniques as Theorem 10.2 of \cite{geer2000empirical} gives us
		\begin{gather*}
			\left\|\meshy-f^*\right\|_n^2 \le O_P\left(\lambda_n\right). 
		\end{gather*} 
		Hence, 
		\[
		\left\|\meshy-f^*\right\|_n^2 = O_p\left(\lambda_n+6\Gamma_{n,m}\right).
		\]
		\qed 
	\end{proof}
	
	\subsection{Interpolation Error on the Mesh}
	
	Consider $x_i \in [0,1]$. Let $D=\left(d_1,\ldots,d_m\right)$ denote the equally spaced grid such that $\delta_m=\diplus-\di=\frac{1}{m}$. Suppose $f\in C^{k+1}[a,b]$. One approach to fitting an observation $x_i$ over the specified grid $D$ is to use a $k$th-order Lagrange interpolating polynomial:
	\begin{gather}
		\Oxi = \sum_{j=0,j\ne i}^{m}f(d_j)L_{m,j}(x_i),
	\end{gather}
	where $L_{m,j}(x_i)=\prod_{j'=0}^{m}\frac{x_i-d_j}{d_j-d_{j'}}$. By Theorem 3.3 of Burden and Faires (2005), 
	\begin{gather}
		f(x_i)=\Oxi+\frac{f^{(k+1)}(\zeta_i)}{(k+1)!}\prod_{j=0}^{m}(x_i-d_j),
	\end{gather}
	where $\zeta_i \in [0,1]$. With $f^{(k+1)}(\zeta_i)<K$, it follows that
	\begin{align}
		|\Oxi-f(x_i)| & = \frac{f^{(k+1)}(\zeta_i)}{(k+1)!}\prod_{j=0}^{m}|x_i-d_j| \\
		& \le \frac{K}{(k+1)!}(m-1)!\delta_m^{k+1} \\
		& \le K'\delta_m^{k+1}.
	\end{align}
	Hence, using a regular grid and a $k$th-order interpolating polynomial for $\Oxi$, 
	\[
	\left|\Oxi-f(x_i)\right| = O\left(m^{-(k+1)}\right).
	\]
	
	\bibliographystyle{apalike}
	\bibliography{myarticles}  
\end{document}